\documentclass[a4paper,UKenglish]{lipicsarxiv}

\usepackage{microtype}

\title{Comparing the Power of Advice Strings: a Notion of Complexity for Infinite Words}

\titlerunning{Comparing the Power of Advice Strings}

\author{Gaëtan Douéneau-Tabot}{\'Ecole Normale Supérieure Paris-Saclay, Université Paris-Saclay, Cachan, France\footnote{This work was partially done during a stay of the author in RWTH Aachen University.}}{gaetan.doueneau@ens-paris-saclay.fr}{}{}

\authorrunning{G. Douéneau-Tabot}

\Copyright{Gaëtan Douéneau-Tabot}

\subjclass{Theory of computation $\rightarrow$ Automata over infinite objects}


\category{}

\relatedversion{ICALP 2018 version available at \cite{doueneau2018}, \href{http://drops.dagstuhl.de/opus/volltexte/2018/9126/}{10.4230/LIPIcs.ICALP.2018.122}}

\supplement{}

\funding{}

\acknowledgements{This work would never have been possible without the support of Erich Grädel. The author also thanks Wolfgang Thomas, Christof Löding, Frédéric Reinhardt and anonymous referees for their advice.}

\EventEditors{}
\EventNoEds{0}
\EventLongTitle{}
\EventShortTitle{}
\EventAcronym{}
\EventYear{}
\EventDate{}
\EventLocation{}
\EventLogo{}
\SeriesVolume{}
\ArticleNo{}

\usepackage{tikz}
\usepackage[utf8]{inputenc}
\usetikzlibrary{shapes,arrows,positioning,automata,shadows}
\usepackage{amsfonts, amsmath, amssymb,dsfont, mathrsfs}
\usepackage{amsthm}
\usepackage{hyperref}
\usepackage[ruled]{algorithm2e}

\newtheorem{proposition}[theorem]{\bfseries Proposition}
\newtheorem{fact}[theorem]{\bfseries Fact}

\newcommand{\mb}[1]{\mathbb{#1}}
\newcommand{\mc}[1]{\mathcal{#1}}
\newcommand{\mf}[1]{\mathfrak{#1}}
\renewcommand{\epsilon}{\varepsilon}
\newcommand{\pc}{\preccurlyeq}

\newcommand{\reg}{\operatorname{Reg}}
\newcommand{\regi}{\reg^{\infty}}
\newcommand{\wreg}{\omega\hspace{-0.07cm}\operatorname{Reg}}
\newcommand{\REG}{\pc_{\reg}}
\newcommand{\REGi}{\pc_{\regi}}
\newcommand{\wREG}{\pc_{\wreg}}
\newcommand{\aut}{\operatorname{AutStr}}
\newcommand{\auti}{\aut^\infty}
\newcommand{\waut}{\omega \hspace{-0.1cm}\operatorname{AutStr}}
\newcommand{\MSO}{\operatorname{\textsf{MSO}}}
\newcommand{\WMSO}{\operatorname{\textsf{WMSO}}}
\renewcommand{\S}{\operatorname{\textsf{S}}}
\newcommand{\FS}{\operatorname{\textsf{FS}}}
\newcommand{\MSOT}{\operatorname{\textsf{MSOT}}}
\newcommand{\FO}{\operatorname{\textsf{FO}}}
\newcommand{\LTL}{\operatorname{\textsf{LTL}}}
\newcommand{\DFT}{\operatorname{\textsf{1WFT}}}
\newcommand{\SST}{\operatorname{\textsf{SST}}}
\newcommand{\WFT}{\operatorname{\textsf{2WFT}}}

\newcommand{\Pref}{\operatorname{Pref}}
\newcommand{\dom}{\operatorname{dom}}
\newcommand{\imp}{\operatorname{Imp}}
\newcommand{\ind}{\operatorname{Index}}
\newcommand{\code}{\operatorname{Code}}
\newcommand{\atoms}{\operatorname{Atoms}}
\newcommand{\qs}{Q_{\atoms}}
\newcommand{\qi}{Q_{\subseteq}}
\newcommand{\val}{\operatorname{\textsf{val}}}

\newcommand{\out}{\operatorname{\textsf{out}}}
\newcommand{\X}{\operatorname{\textsf{X}}}
\newcommand{\G}{\operatorname{\textsf{G}}}
\newcommand{\U}{\operatorname{\textsf{U}}}
\renewcommand{\S}{\operatorname{\textsf{S}}}

\begin{document}

\maketitle

\begin{abstract}

This paper is the extended version of On the Complexity of Infinite Advice Strings \cite{doueneau2018}.

We investigate a notion of comparison between infinite strings. In a general way, if $\mc{M}$ is a computation model (e.g. Turing machines) and $\mc{C}$ a class of objects (e.g. languages), the complexity of an infinite word $\alpha$ can be measured with respect to the amount of objects from $\mc{C}$ that are presentable with machines from $\mc{M}$ using $\alpha$ as an oracle.

In our case, the model $\mc{M}$ is finite automata and the objects $\mc{C}$ are either recognized languages or presentable structures, known respectively as advice regular languages and advice automatic structures. This leads to several different classifications of infinite words that are studied in detail; we also derive logical and computational equivalent measures. Our main results explore the connections between classes of advice automatic structures, $\MSO$-transductions and two-way transducers. They suggest a closer study of the resulting hierarchy over infinite words.
 \end{abstract}



\section{Introduction}

Several measures have been defined to describe the (intuitive) complexity of infinite strings. Among others we mention subword complexity \cite{allouche2003}, Kolmogorov complexity, and Turing degrees \cite{sacks1963}. Whereas the two first methods focus on the intrinsic information contained in a string, the other one studies the relation of computability from one word to another, defining a preorder whose properties are now quite well understood. Equivalently, this preorder compares the expressive power of Turing machines that use an infinite word as oracle.

This paper follows a similar idea: we consider finite automata that can access an infinite \emph{advice} string while processing their input. Such automata define classes of \emph{advice regular languages} \cite{salomaa1968}, that generalize standard regularity. This notion enables us to introduce a comparison for infinite words: $\alpha$ is simpler (in the sense of languages) than $\beta$ if every language recognized by an automaton with advice $\alpha$ can also be recognized with advice $\beta$, what corresponds to the intuition that $\alpha$ contains less information than $\beta$.

Before going further, we evoke the current motivations around advice regular languages. Standard regular languages can be used to encode finite signature structures, known as \emph{automatic structures}. This concept, derived from Büchi's early automata-logic techniques, has been shown especially relevant since its formalization in \cite{khoussainov1995} and \cite{blumensath2000}. The model opened the door to a vast range of decision procedures via automata constructions, but it suffers from a lack of expressiveness, since e.g. $\langle \mb{Q}, +\rangle$ is not automatic \cite{tsankov2011}. However, $\langle \mb{Q}, +\rangle$ is an example of \emph{advice automatic structure}: it can be encoded using advice regular languages (instead of regular languages) \cite{KRSZ12}. Such structures share many properties with the former automatic structures, furthermore the use of advices builds an interesting framework to discuss algorithmic meta-theorems \cite{abu2017}. We shall not follow a model-theoretic point of view on advice automatic structures, but we use them to define another notion of comparison over infinite words as follows: $\alpha$ is simpler (in the sense of structures) than $\beta$ if every  automatic structure with advice $\alpha$ is also  automatic with advice $\beta$.

\textbf{Objectives and outline.} This paper is structured as a quest for a relevant way to compare infinite strings through the notion of advice. The informal criteria we use to define a good complexity measure are the following: it should have a simple definition, be robust enough, but not too coarse because we want to separate understandable sequences. Note that Turing degrees do not match this intuition since they make no distinction between computable -  useful - sequences. Our results will establish an interesting correspondence between the expressive power of advices (compared more or less using \emph{languages}) and certain forms of \emph{transductions}, when considering the way they classify infinite strings. This is somehow surprising, since the theory of transformations between words tends to be more fruitful and more difficult than the study of languages, following an early remark of Dana Scott \cite{scott1967}: ``the functions computed by the various machines are more important - or at least more basic - than the sets accepted by these device''. The concept of advice helps unifying these two frameworks. Furthermore, we shall use this idea to provide slightly new perspectives on (advice) automatic presentations and logic over infinite words.

After recalling preliminary results on structures, formal languages and logic, we present formally in Section \ref{sec:reglang} the notion of regularity with advice, under several variants. We study the comparisons of words provided by classes of advice  regular languages, as evoked above. An easy correspondence is drawn with transductions, for instance we show that every regular language with advice $\alpha$ is also regular with advice $\beta$ if and only if $\alpha$ is the image of $\beta$ under a Mealy machine. Nevertheless, we note that such comparisons are far from being robust. Next we turn in Section \ref{sec:autstr} to classes of advice  automatic structures and briefly study some standard properties. We then show that some variants of advice regular languages have no influence on the classes of presentable structures. This first involved result is also a first step to obtain a new robust notion of comparison. The proof of this result also provides an original normal form for $\MSO$-formulas when interpreted in a fixed word model.

Section \ref{sec:msot} intends to understand the comparison over infinite words defined with respect to advice automatic presentations (see above); it develops the most involved results of this paper. Similar investigations were built in \cite{loding2007}, under the formalism of set-interpretations - a very close notion. We particularize their results to show that every  automatic structure with advice $\alpha$ is also  automatic with $\beta$ if and only if $\alpha$ is the image of $\beta$ under an $\MSO$-transduction (some logical transformation between words). We then give a more handy equivalent statement: $\alpha$ is the image of $\beta$ under a two-way transducer. This result is quite specific and original, since such transducers are however not powerful enough to realize all \emph{functions} of infinite words defined by $\MSO$-transductions \cite{alur2012}. We conclude this paper investigating more precisely the structure of this last preorder (defined in particular by two-way transductions) in Section \ref{sec:hierarchy}. Even if no previous research was done on the subject, a very similar study was carried out in \cite{endrullis2015} for comparison via one-way finite transducers. In the light of their results, we rough out the structure of a new hierarchy and explain why an involved questioning may be fruitful.

\section{Preliminaries}

\paragraph*{Words and languages}

Greek capitals $\Sigma$, $\Gamma$ and $\Delta$ are used to denote alphabets, i.e. finite sets of letters; $\square$ is a padding letter that never belongs to these alphabets. If $w$ is a (possibly infinite) word, let $|w| \in\mb{N}\cup\{\omega\}$ be its length, and for $n \ge 0$ let $w[n]$ be its $(n+1)$-th letter (when defined). For $0 \le m \le n$, let $w[m:n] = w[m]w[m+1] \cdots w[n-1]$ (when defined, possibly $\epsilon$). We write $w[:n]$ for the prefix $w[0:n]$, and $w[n:]$ for the (possibly infinite) suffix $w[n] w[n+1] \cdots $.

Denote by $\reg$ (resp. $\wreg$) the class of regular (resp. $\omega$-regular) languages. We assume familiarity with the standard results of automata theory. We make a large use of logic-automata connections, especially over infinite words (see e.g. \cite{thomas1997} for a good survey).

\begin{definition}[Convolution] If $u$ and $v$ are (possibly infinite) words, their \emph{convolution} $u \otimes v$ is the word of length $\max(|u|,|v|)$ such that:
\begin{itemize}
\item $(u\otimes v)[n] = (u[n],v[n])$ if $n < \min(|u|,|v|)$;
\item $(u\otimes v)[n] =(u[n],\square)$ if $|v| \le n < |u|$;
\item $(u\otimes v)[n] = (\square,v[n])$ if $|u| \le n < |v|$.
\end{itemize}

\end{definition}

Convolution is defined in a similar way for $k$-tuples of (finite or infinite) words.

\paragraph*{Structures and logic}

We deal with structures over a finite signature, denoted by fraktur letters $\mf{A}$, $\mf{B}$\dots{} When needed, structures are seen as purely relational (we replace the functions by their graphs). We associate to each infinite word $\alpha$ its word structure $\mf{W}^\alpha = \langle \mb{N}, <, (P_a)_{a \in \Gamma}\rangle$ where $<$ is the usual ordering on positive integers, and $n \in P_a$ if and only if $\alpha[n] = a$. For succinctness reasons, $\alpha \models \phi$ often stands for $\mf{W}^\alpha \models \phi$. If $\tau$ is a signature and $\mc{L}$ a logic, $\mc{L}[\tau]$-formulas are $\mc{L}$-formulas over the signature $\tau$. We assume that equality implicitly belongs to every signature and write $\MSO[<,\Gamma]$ for $\MSO[<,(P_a)_{a \in \Gamma}]$. $\MSO$-formulas can be interpreted using weak semantic ($\WMSO$), where we allow set quantifications to range only over finite sets.

We recall how to present structures with languages: we encode the elements of the domain as words, so that the relations can be described in consistent way. 

\begin{definition}\label{def:pres}

Let $\mf{A} := \langle A, R_1 \dots R_n \rangle$ a relational structure and $\mc{C}$ a class of languages (possibly over infinite words). A  \emph{$\mc{C}$-presentation} of $\mf{A}$ is a tuple $(L, L_{=}, L_1 \dots L_n)$ of languages from $\mc{C}$ such that there exists a surjective function $\nu : L \rightarrow A$ with:
\item 
\begin{itemize}

\item $L_= = \{w \otimes w'~|~w,w' \in L \text{ and } \nu(w)  = \nu(w')\}$;
\item for $R_i$ (arity $r_i$), $L_i = \{w_1\otimes \cdots \otimes w_{r_i}~|~\forall 1 \le j \le r_i, w_j \in L \text{ and } (\nu(w_1), \dots, \nu(w_{r_i})) \in R_i\}$.
\end{itemize}

\end{definition}

The function $\nu$ describes how $A$ is encoded in $L$. Since we never deal with the elements directly, it does not belong explicitly to the presentation and can be considered as a notation. The alphabet of $L$ is called \emph{encoding alphabet} and often denoted $\Sigma$. The presentation is said \emph{injective} if $L_= := \{w \otimes w~|~w \in L\}$.

The point is to find a class of languages which is both robust and decidable enough. If $\mc{C}$ is the class of recursive languages, the $\mc{C}$-presentable structures correspond to early studied recursive structures \cite{millar1978}. More recently, the class of $(\omega)\hspace{-0.1cm}\reg$-presentable structures generated much attention, under the name of \emph{($\omega$-)automatic structures} \cite{blumensath2000}. Such structures can be described using a tuple of automata for the languages of the presentations.We denote by $(\omega)\hspace{-0.1cm}\aut$ the class of  ($\omega$-)automatic structures.

\begin{example} $\langle\mb{N}, +, 0, 1\rangle \in \aut$.

\end{example}

The well-known behavior of automata produces literally hundreds of nice properties in this field, the most famous being probably the following.
 
 \begin{proposition}[\cite{blumensath2000}]
 
 Every $(\omega)$-automatic structure has a decidable $\FO$-theory. The method is effective starting from a presentation by finite automata.
 
 \end{proposition}

 One of the main current challenges is to describe which structures have an automatic presentation, and elegant characterizations have been stated for certain classes, such as finitely generated groups \cite{oliver2005}. However, as shown in Theorem \ref{theo:rat}, the presentation fails for simple structures with decidable $\FO$-theory. This motivates the study of possible extensions.

\begin{theorem}[\cite{tsankov2011}]

\label{theo:rat}

$\langle \mb{Q}, + \rangle$ is not an ($\omega$-)automatic structure.

\end{theorem}

\paragraph*{Interpretations}

A useful tool in model theory is the concept of \emph{interpretation}, describing a structure in another (host) structure via a tuple of logical formulas.

\begin{definition}[interpretation] Let $\mf{A}$ be a structure over a signature $\tau$, $\mc{L}$ be a logic and  $\mc{I} := (\phi_\delta(\overline{x}),\phi_=(\overline{x},\overline{y}), \phi_1(\overline{x_1} \dots \overline{x_{r_1}}) \dots \phi_{p}(\overline{x_1} \dots \overline{x_{r_{p}}}) )$ a tuple of $\mc{L}[\tau]$-formulas where $\overline{x},\overline{y}$ and the $ \overline{x_i}$ are $k$-tuples of free variables. Let

\item

\begin{itemize}

\item $A_\delta := \{\overline{a} = (a_1 \dots a_k)~|~\mf{A} \models \phi_\delta(\overline{a})\}$;

\item $\sim$ a binary relation on $A_\delta$ with $\overline{a} \sim \overline{b}$ if and only if $\mf{A} \models \phi_=(\overline{a},\overline{b})$;

\item for $1 \le i \le p$, $R_i$ is a relation on $A_\delta$ defined as $(\overline{a_1} \dots \overline{a_{r_i}}) \in R_i$ if and only if $\mf{A} \models \phi_i(\overline{a_1} \dots\overline{a_{r_i}})$.

\end{itemize}

We say that $\mc{I}$ is a \emph{$k$-dimensional $\mc{L}$-interpretation} of a structure $\mf{B}$ in the structure $\mf{A}$ if the following conditions are met:

\item 

\begin{itemize}

\item $\sim$ defines an congruence relation on $A_\delta$ with respect to $R_1 \dots R_{p}$;

\item $\langle A_\delta, R_1 \dots R_{p} \rangle / \sim$ is isomorphic to $\mf{B}$.
\end{itemize}

\end{definition}

The interpretation is said \emph{injective} if $\sim$ is the equality relation of $A_\delta$. In the literature, interpretations are often directly assumed to be \emph{1-dimensional injective interpretations}. The choice of the logic $\mc{L}$ gives several kinds of interpretation, detailed in Definition \ref{def:zoointer}.

\begin{definition} \label{def:zoointer}

\item 
\begin{enumerate}

\item An \emph{$\FO$-interpretation} is a tuple of $\FO$-formulas. The elements of $\mf{A}$ are encoded as tuples of elements in the host structure $\mf{B}$.

\item An \emph{$\MSO$-interpretation} is a tuple of $\MSO$-formulas with free first-order variables. If we use the weak semantic, we speak of $\WMSO$-interpretation. Once more, the elements of $\mf{A}$ are encoded as tuples of elements of $\mf{B}$.

\item An \emph{$\S$-interpretation} (set) is a tuple of $\MSO$-formulas with free set variables. If we use weak semantic, we speak of \emph{$\FS$-interpretation} (finite set). The elements of $\mf{A}$ are encoded as tuples of (finite) sets of elements in the host structure.

\end{enumerate}

\end{definition}

We briefly recall the behavior of interpretations with respect to composition.

\begin{fact}[closure under composition] \label{fact:compo} \item
\begin{enumerate}

\item If $\mf{A}$ is $\FO$-interpretable in $\mf{B}$ which is $\FO$-interpretable in $\mf{C}$, then $\mf{A}$ is directly $\FO$-interpretable in $\mf{C}$.

\item If $\mf{A}$ is $\MSO$-interpretable in $\mf{B}$ which is 1-dimensionally $\MSO$-interpretable in $\mf{C}$, then $\mf{A}$ is directly $\MSO$-interpretable in $\mf{C}$
\end{enumerate}
\end{fact}

\begin{proof}[Proof idea (folklore).] The formulas of the interpretation in $\mf{B}$ can equivalently be described in $\mf{A}$, while adding some variables if necessary.

\end{proof}

The presence of sets and the use of several dimensions allows to describe more transformations, but it forces to be careful in the statements of the previous fact:
\begin{itemize}

\item If $\mf{A}$ is (1-dim.) $\S$-interpretable in $\mf{B}$ which is (1-dim.) $\S$-interpretable in $\mf{C}$, there is no reason why $\mf{A}$ should be $\S$-interpretable in $\mf{C}$. Indeed, the elements of $\mf{C}$ (i.e. sets of $\mf{B}$) should be coded as sets of sets of $\mf{A}$.

\item The case of $\MSO$ is a bit more subtle. If $\mf{A}$ is $\MSO$-interpretable in $\mf{B}$ which is $k$-dim. $\MSO$-interpretable in $\mf{C}$ with $k >1$, there is no reason why $\mf{A}$ should directly be $\MSO$-interpretable in $\mf{C}$. Indeed, the sets of $\mf{B}$ (that can be used in the last interpretation) are sets of $k$-tuples of elements of $\mf{C}$, but we can only describe $k$-tuples of sets, what is formally different.

\end{itemize}

\begin{remark} The composition properties allow - in specific cases - to transfer the decidability of the logical theory from the host structure to the other.
\end{remark}

Interpretations are a key concept to extend standard automata-logic equivalences to automatic structures.

\begin{proposition}[\cite{khoussainov2004}] A structure $\mf{A}$ is automatic (resp. $\omega$-automatic) if and only if $\mf{A}$ is $\FS$-interpretable (resp. $\S$-interpretable) in $( \mb{N}, < )$. 

\end{proposition}

\section{Simple case: regular languages with advice}

\label{sec:reglang}

We present in this section an extension of regular languages known as \emph{regular languages with advice}. This concept enables us to study some preorders over infinite words; we discuss their relevance and establish a first link with transductions. The fruits we catch here are hanging close to the ground, but they raise intuitions about the climbing that follows.

\subsection{Terminating languages}

\label{sec:term}

The idea of advice regularity is to consider languages accepted by automata that read an infinite advice string while processing its input \cite{baer1969}. We provide an equivalent definition which does not directly deal with automata but only languages.

\begin{definition}

\label{def:regad}

$L \subseteq \Sigma^*$ is \emph{terminating regular} with advice $\alpha \in \Gamma^\omega$ if there exists a regular language $L' \subseteq (\Sigma \times \Gamma)^*$ such that $L = \{ w~|~w \otimes \alpha[:|w|] \in L')\}$.

\end{definition}

\begin{example}

\label{ex:reg1} \item

\begin{enumerate}

\item If $L \subseteq \Sigma^*$ is regular, so is $\{w \otimes w'~|~w \in L, w' \in \Gamma^*, |w| = |w'|\}$, and considering this language shows that $L$ is regular with any advice;

\item the set $\Pref(\alpha) := \{\alpha[:n]~|n \ge 0\}$ is regular with advice $\alpha$.
 
 \end{enumerate}
\end{example}

\begin{remark} There are non-computable languages regular with some advice.

\end{remark}

We denote by $\reg[\alpha]$ the class of regular languages with advice $\alpha$. As evoked in the introduction, our goal is to measure the complexity of infinite words, through the expressiveness of their advice classes. We write $\alpha \REG \beta$ whenever $\reg[\alpha] \subseteq \reg[\beta]$, this relation is clearly a preorder over infinite words. Let the \emph{$\REG$-degrees} be the equivalence classes of the relation $\REG \cap \succcurlyeq_{\reg}$, they describe the sets of equally complex advices. We remark that ultimately periodic words (i.e. infinite words of the form $u v^\omega$) form the least $\REG$-degree; indeed the inclusion $\reg \subseteq \reg[\alpha]$ is strict if and only if $\alpha$ is not ultimately periodic \cite{baer1969, reinhardt2013}. We now provide a first equivalence with transductions.

\begin{definition}

A (deterministic) \emph{Mealy machine} is a 6-tuple $( Q,q_0,  \Delta, \Gamma, \delta, \theta )$ where $Q$ is the finite set of states, $q_0 \in Q$ is the initial state, $\Delta$ is the input alphabet, $\Gamma$ is the output alphabet, $\delta: Q \times \Delta \rightarrow Q$ is the (partial) transition function, $\theta: Q \times \Delta \rightarrow \Gamma$ is the (partial) output function.

\end{definition}

A run of a Mealy machine is a run of the underlying deterministic automaton. On input $\beta$, the machine outputs $\alpha$ the concatenation of the outputs along the run on $\beta$.

\begin{proposition} \label{prop:carreg} The following conditions are equivalent:

\item

\begin{enumerate}

\item $\reg[\alpha] \subseteq \reg[\beta]$;

\item $\alpha$ is the image of $\beta$ under some Mealy machine.

\end{enumerate}

\end{proposition}

\begin{proof}
A Mealy machine answering $\alpha$ on $\beta$ clearly provides a way to transform any language of $\reg[\alpha]$ into a language of $\reg[\beta]$, after some composition. Conversely, if $\reg[\alpha] \subseteq \reg[\beta]$ then  $\Pref(\alpha)= \{ w~|~w \otimes \beta[:|w|] \in \mc{L}(\mc{A})\}$ for some deterministic automaton $\mc{A}$ on $\Gamma \times \Delta$. Since the (unique) run on a finite word $\alpha[:n] \otimes \beta[:n]$ only uses accepting states (due to determinism and prefix-closure), non-accepting states can wlog be removed. If the resulting automaton has transitions of the form $q \rightarrow^{(a,b)} q'$ and $q \rightarrow^{(a',b)} q''$ with $a \neq a'$, we can remove them. Indeed, it is impossible for a run on some $\alpha[:n] \otimes \beta[:n]$ to use one of them, since all states are now accepting. This last automaton can easily be seen as a Mealy machine.

\end{proof}

Comparison via $\REG$ thus corresponds to computability via Mealy machines. The properties of this preorder were studied under this form in \cite{belovs2008}. However, tiny changes in the words completely modify their $\REG$-degree: those classes are far from being robust.

\begin{fact}[\cite{belovs2008}] \label{fact:chain}
Whenever $\alpha$ is not ultimately periodic, we have a strictly increasing chain $\alpha \prec_{\reg} \alpha[1:] \prec_{\reg} \cdots \prec_{\reg} \alpha[n:] \prec_{\reg} \cdots$. A strictly decreasing chain can be obtained similarly with $\alpha \succ_{\reg} \square \alpha \succ_{\reg}  \dots \succ_{\reg}  \square^n\alpha \succ_{\reg} \cdots $.

\end{fact}

An interesting point is the closure properties of these classes.

\begin{proposition}[\cite{baer1969}] $\reg[\alpha]$ is closed under boolean operations.

\end{proposition}

However, when $\alpha$ is not ultimately periodic, $\reg[\alpha]$ is not closed under projection (with respect to $\otimes$) \cite{reinhardt2013}. This is a serious issue if one intends to encode logical theories, what may explain why automata with advice have remained unused for many years. A possible solution, detailed in the next paragraph, is to use $\omega$-regularity instead of finite regularity.

\subsection{Non-terminating languages and $\omega$-regularity}

Once more, we shall provide a definition in terms of languages, but it could equivalently be defined with $\omega$-automata that read an advice string.

\begin{definition}[\cite{KRSZ12}]

$L \subseteq \Sigma^\omega$ is $\omega$-regular with advice $\alpha \in \Gamma^\omega$ if there is an $\omega$-regular language $L' \subseteq (\Sigma \times \Gamma)^\omega$ such that $L = \{ w~|~w \otimes \alpha \in L')\}$.

\end{definition}

\begin{example} \item

\begin{enumerate}

\item Every $\omega$-regular language is also $\omega$-regular with any advice;

\item $\{\alpha\} $ is $\omega$-regular with advice $\alpha$.

\end{enumerate}
\end{example}

We denote by $\wreg[\alpha]$ the class of $\omega$-regular languages with advice $\alpha$. The next definition generalizes $\omega$-regularity with advice to finite-words languages.

\begin{definition}[\cite{KRSZ12}]

A language $L \subseteq \Sigma^*$ is non-terminating regular with advice $\alpha \in \Gamma^\omega$ if there is an $\omega$-regular language $L' \subseteq ((\Sigma \uplus \square) \times \Gamma)^\omega$ such that $L = \{w~|~w \otimes \alpha \in L'\}$.

\end{definition}

\begin{example} \label{ex:reginf} $\forall n \ge 0$, $\Pref(\alpha[n:])$ is non-terminating regular with advice $\alpha$.

\end{example}

Let $\reg^{\infty}[\alpha]$ be the class of non-terminating regular languages with advice $\alpha$. It follows from the definitions that $L \in \reg^{\infty}[\alpha] $ if and only if $ \{w\square^\omega~|~w \in L\} \in \wreg[\alpha]$. These new definitions increase the expressiveness of advice languages, since $\reg[\alpha] \subseteq \reg^\infty[\alpha]$ and the inclusion is strict when $\alpha$ is not ultimately periodic \cite{KRSZ12}. Furthermore, they solves the lack of closure properties evoked in the end of Subsection \ref{sec:term}.

\begin{proposition}[\cite{KRSZ12}] \label{prop:closewreg} $\reg^\infty[\alpha]$ and $\wreg[\alpha]$ are closed under boolean operations, cylindrification, and projection (with respect to $\otimes$).

\end{proposition}

Let us compare infinite words with respect to this $\omega$-regular use of advice. We define the preorders $\REGi$ (resp. $\wREG$) based on the inclusion of the $\reg^{\infty}$ (resp. $\wreg$) classes, and the corresponding notions of degrees. It is not hard to see that  ultimately periodic words are again the least $\REGi$- and $\wREG$-degree. We now make a non-trivial step towards a generic correspondence between advices, machine transductions, and logic.

\begin{definition} An $\omega$-regular function $f$ is a (partial) mapping $\Gamma^\omega \rightarrow \Delta^\omega$ whose graph $\{w \otimes f(w)~|~w \in \dom(f)\}$ is an $\omega$-regular language.

\end{definition}

\begin{definition} \label{def:relabel}We say that $\alpha \in \Gamma^\omega$ is the image of $\beta \in \Delta^\omega$ under an \emph{$\MSO$-relabelling} if there is a tuple $\MSO[<, \Delta]$-formulas $(\phi_a(x))_{a \in \Gamma}$ such that: $\forall n \ge 0$, $\alpha[n] = a$ if and only if $\beta\models \phi_a(n)$.

\end{definition}

\begin{proposition} \label{prop:carwreg} For $\alpha \in \Gamma^\omega$ and $\beta \in \Delta^\omega$, the following are equivalent:

\item

\begin{enumerate}

\item $\reg^\infty[\alpha] \subseteq \reg^\infty[\beta]$;

\item $\wreg[\alpha] \subseteq \wreg[\beta]$;

\item $\alpha$ is the image of $\beta$ under some $\omega$-regular function;

\item $\alpha$ is the image of $\beta$ under some $\MSO$-relabelling.

\end{enumerate}

\end{proposition}

\begin{proof} See Appendix \ref{proof:prop:carwreg}. \end{proof}

\begin{remark}  A word $\alpha$ is the image of $\beta$ under some Mealy machine if and only if $\alpha$ is the image of $\beta$ under a \emph{relativized $\MSO$-relabelling}, defined as a relabelling where in the formulas $\phi_a(x)$ every quantification is relativized under $x$, i.e. of the form $Q y/Y \le x$, see Fact \ref{fact:MMmso}.

\end{remark}

We obtain in particular $\wREG = \REGi$ and $\REG \subsetneq \REGi$ (see Fact \ref{fact:chain} and Example \ref{ex:reginf}). To understand its structure, we briefly give a simple necessary condition for $\alpha\REGi \beta$.

\begin{definition}

Let $\alpha \in \Gamma^\omega$. Its \emph{subword complexity} is the function $p_{\alpha}: \mb{N} \rightarrow \mb{N}$  defined by $p_\alpha(k) = \# \{w \in \Gamma^k~|~w \text{ is a factor of } \alpha\}$.

\end{definition}

This function counts for each $k \ge 0$ the number of factors of size $k$ appearing in $w$. We now show that this measure can only decrease when applying an $\omega$-regular function.

\begin{proposition}

\label{prop:subword}

If $\alpha \in \Gamma^\omega$ is the image of $\beta \in \Delta^\omega$ under some $\omega$-regular function, then $p_\alpha \le K \times p_\beta$ for some constant $K \ge 0$.

\end{proposition}

\begin{proof}
 
Let $\mc{A}$ be an $\omega$-automaton with states $Q$,  that describes the graph of the function in $(\Gamma \times \Delta)^\omega$, then $\{\alpha\} = \{w~|~w \otimes \beta \in \mc{L}(\mc{A})\}$. Let $\rho$ be an accepting run of $\mc{A}$ on $\alpha \otimes \beta$, $\rho(m)$ being state after reading letter $m$. If $k,m,m' \ge 0$ are such that $\beta[m:m+k] = \beta[m':m'+k]$, $\rho(m-1) = \rho(m'-1)$, and $\rho(m+k) = \rho(m'+k)$, then both $\alpha \otimes \beta $ and $(\alpha[:m] \alpha[m':m'+k] \alpha[m+k+1:])\otimes \beta$ are accepted. Thus $\alpha[m:m+k]$ and $\alpha[m':m'+k']$ must be equal. A pigeonhole argument then shows that there are at most $|Q|^2  \times p_{\beta}(k)$ factors of size $k$ in the word $\beta$.

\end{proof}

For all $n \ge 1$, there exists a (computable) string $\alpha_n$ such that $p_{\alpha_n} : k \mapsto n^k$. Necessarily $\regi[\alpha_n]$ is not contained in any $\regi[\beta]$ for $\beta \in \{1, \dots, n-1\}^{\omega}$ because $p_{\beta}(k) \le ({n-1})^k$. This observation shows that the size of the alphabet is an unavoidable parameter for $\REGi$, which is not good news when looking for a robust notion of complexity. The rest of this paper will no longer deal with the preorders defined by languages, but it move towards presentable structures in order to describe a more relevant notion of comparison.

%
%

\section{Advice automatic structures}

\label{sec:autstr}

An interesting point of the previous results is the correspondence described by Propositions \ref{prop:carreg} and \ref{prop:carwreg}: they relate the power of advices to logical or computational comparisons. However, the resulting preorders were somehow disappointing.

We now consider structures that are presentable with advice in order to derive similar results. Following the definitions of \cite{abu2017}, we denote by $\aut[\alpha]$ the class of $\reg[\alpha]$-presentable structures, $\aut^\infty[\alpha]$ for $\reg^\infty[\alpha]$-presentable, and $\waut[\alpha]$ for  $\wreg[\alpha]$-presentable (see Definition \ref{def:pres}). Such structures are said (\emph{$\omega$-)automatic with advice $\alpha$}. Their study is located a level of abstraction higher than what was done in Section \ref{sec:reglang}, since the languages have no longer importance in theirselves, but are only used to encode other objects.

\subsection{Tools and basic properties of advice presentations}

An advice automatic structure can be described ``effectively'' via a tuple of automata (as for standard automatic structures), and a certain advice $\alpha$. In fact, the decidability feature of automatic structures is preserved as soon as $\alpha$ is decidable enough.

\begin{proposition}[\cite{abu2017}] 
 If $\mf{W}^\alpha$ has a decidable $\MSO$-theory, every structure in $\waut[\alpha]$, $\auti[\alpha]$ or $\aut[\alpha]$ has a decidable $\FO$-theory.

\end{proposition}

Large classes of infinite words with decidable $\MSO$-theory have been described, see e.g. \cite{barany2008} or \cite{semenov1984}. We briefly show why the generalization from automatic structures to advice automatic structures can be fruitful (compare the next result to Theorem \ref{theo:rat}).

\begin{fact}[\cite{KRSZ12}]

$\langle \mb{Q}, +\rangle \in \aut[\alpha]$ for some advice $\alpha$ with decidable $\MSO$-theory.

\end{fact}

We now briefly describe basic properties of presentations with advice.

\begin{fact} Inclusion of language classes give $\aut \subseteq \aut[\alpha] \subseteq \auti[\alpha] $ and $\waut \subseteq \waut[\alpha]$. Inclusions are equalities if $\alpha$ is ultimately periodic.

\end{fact}

\begin{remark} There is however no immediate argument to deduce $\aut \subsetneq \aut[\alpha]$ when $\alpha$ is not ultimately periodic. We shall see in Section \ref{sec:hierarchy} that this statement is true.

\end{remark}

As an immediate consequence of the definitions, $\auti[\alpha] \subseteq \waut[\alpha]$ and $\waut[\alpha]$ contains uncountable structures, whereas $\auti[\alpha]$ does not. This idea can be refined.

\begin{theorem}[\cite{abu2017}]

\label{theo:countable}

$\aut^\infty[\alpha]$ is exactly the subclass of countable structures of $\waut[\alpha]$.

\end{theorem}

The next result shows to what extent the advice contains the seeds of every presentation, and how we generalized the case of automatic structures.

\begin{proposition}[\cite{zaid2016}]
\label{prop:si}

\item
\begin{enumerate}

\item $\mf{A} \in \waut[\alpha]$ if and only if $\mf{A}$ is $\S$-interpretable in $\mf{W}^\alpha$;

\item $\mf{A} \in \aut^\infty[\alpha]$ if and only if $\mf{A}$ is $\FS$-interpretable in $\mf{W}^\alpha$.

\end{enumerate}
\end{proposition}

\begin{remark}[\cite{zaid2016}] If the presentation is injective and the encoding is alphabet binary, the resulting interpretation can be done $1$-dimensional and injective.
\end{remark}

We now discuss a few structural properties of advice automatic structures. The statements are not deeply technical nor enlightening, but they are essential tools in the discussions of Section \ref{sec:msot}. A first question is to know whether each presentation can be made injective.

\begin{proposition}[\cite{reinhardt2013}] \label{prop:injective}

If $\mf{A}$ has a $\regi[\alpha]$-presentation, it has an injective $\regi[\alpha]$-presentation.

\end{proposition}

A second point it to understand how the encoding alphabet can be restricted. Binary resentations are enough to describe all automatic structures \cite{blumensath2000}. We show that it is still possible here, up to a small modification of the advice.

\begin{definition}

\label{def:duplication}

For $n \ge 1$ let $\mu_n : \Gamma \rightarrow \Gamma^*$ mapping each letter $a$ to $a^n$. We extend this function to infinite words in a morphic way.
\end{definition}

\begin{example}

$\mu_3((01)^\omega) = (000111)^\omega$.

\end{example}

\begin{proposition}

\label{prop:binary}

If $\mf{A}$ has a $\regi[\alpha]$-presentation, there is $n \ge 1$ such that $\mf{A}$ has a $\regi[\mu_n(\alpha)]$-presentation over a binary encoding alphabet. If the first presentation was injective, so in the second.

\end{proposition}

\begin{proof}[Proof sketch] Let $\mf{A} = \langle A, R_1, \dots R_n \rangle\in \auti[\beta]$ and $(L, L_{=}, L_1, \dots, L_n)$ the corresponding presentation over an alphabet $\Sigma = \{a_1 \dots a_n\}$. The idea is to replace $a_i$ by a binary string of length $n$. Formally let $w_i = 0^k 1^{n-k} $ and let $f : \Sigma \rightarrow \{0,1\} $ mapping $a_i$ to $w_i$. $f$ is extended morphically to (convolutions of) words of $\Sigma^*$. Note $|f(w)| = n |w|$. We check that $(f(L), f(L_{=}), f(L_1), \dots, f(L_n))$ is a tuple of languages of $\regi[\mu_n(\beta)]$ which is still a presentation of $\mf{A}$. If $L_{=} = \{w \otimes w~|~w \in L\}$ then $f(L_{=}) =  \{w \otimes w~|~w \in f(L)\}$ thus injectivity is preserved by this construction.
\end{proof}

\begin{remark} This proof also works for $\reg[\alpha]$- and $\wreg[\alpha]$-presentations.

\end{remark}

\subsection{Terminating and non-terminating encodings}

Dealing directly with $\reg[\alpha]$-presentations seems more difficult, since basic properties lack to this class of languages.  We now show $\auti[\alpha] = \aut[\alpha]$, hence the expression ``advice automatic structure'' is not ambiguous. To give an intuition of the proof, we note that an $\omega$-automaton performs an infinite run on $w \otimes \alpha$ (for $w$ finite) in two steps: first, it follows a finite run on $w \otimes \alpha[:|w|]$, then it checks some $\omega$-regularity on $\square^\omega \otimes \alpha[|w|:]  \simeq \alpha[|w|:]  $.  Basically, the  $\omega$-regularity feature is only used on suffixes of the advice. On the other hand, a automaton for $\reg[\alpha]$ is blind to the $\omega$-future. We show that it can nevertheless look at some ``finite amount of future'' and deduce corresponding $\omega$-regularity on the suffixes. A key idea is that since the advice is \emph{fixed}, so are several properties of its suffixes.

\begin{theorem}

\label{theo:infini}

Let $L$ be an $\omega$-regular language and $\alpha \in \Gamma^\omega$ a fixed word. There is a (finite words) regular language $L'$ and $N \ge 0$ such that for all $n \ge N$, $\alpha[n:] \in L$ if and only if $\alpha[n:]$ has a finite prefix in $L'$. Furthermore, if $L$ can be described by an $\FO[<, \Gamma]$-sentence, $L'$ can be described by an $\FO[<, \Gamma]$-sentence as well.

\end{theorem}

\begin{proof}[Proof]

Both proofs are detailed in Appendix \ref{proof:lem:restrict}. The case of $\FO$ is achieved via expressive equivalence with $\LTL$ (known as Kamp's Theorem, see \cite{rabinovich2014}). For $\MSO$ in general, we make use of results of A.L. Semenov \cite{semenov1984}.

\end{proof}

Corollary \ref{cor:positions} will formalize our intuition that terminating automata can check $\omega$-regular properties on suffixes. It thus enables us to explicit the relationships between $\reg[\alpha]$ and $\reg^\infty[\alpha]$, and between $\aut[\alpha]$ and $\auti[\alpha]$.

\begin{corollary}

\label{cor:positions}

Let $L\subseteq \Gamma^\omega$ be an $\omega$-regular language and $\alpha \in \Gamma^\omega$. There is a function $f: \mb{N} \rightarrow \mb{N}$ such that $\{0^n \square^{f(n)}~|~\alpha[n:] \in L\} \in \reg[\alpha]$.

\end{corollary}

\begin{proof} By applying Theorem \ref{theo:infini} we get a regular language $L'$ and $N \ge 0$ such that for all $n \ge N$, $\alpha[n:] \in L$ if and only if $\alpha[n:]$ has a finite prefix in $L'$. If $\alpha[n:] \in L$, let $f(n)$ be the length of the smallest prefix of $\alpha[n:]$ belonging to $L'$. We take $f(n)$ arbitrarily in the other cases to define a mapping $f : \mb{N} \rightarrow \mb{N}$. The set $\{0^n \lozenge^{f(n)} ~|~n \ge N \text{ and }\alpha[n:n+f(n)+1] \in L'\} = \{0^n \lozenge^{f(n)} ~|~n \ge N \text{ and }\alpha[n:] \in L\}$ is clearly terminating regular with advice $\alpha$. Thus $\{0^n \lozenge^{f(n)} ~|~n \ge 0 \text{ and }\alpha[n:] \in L\} \in \reg[\alpha]$ as well (we hardcode in the automaton what happens before $N$).

\end{proof}

\begin{corollary}

\label{prop:langreg}

Let $\alpha \in \Gamma^\omega$. For every language $L \in \reg^\infty[\alpha]$, there is a function $f: \mb{N} \rightarrow \mb{N}$ such that $\{w\square^{f(|w|)}~|~w\in L\} \in \reg[\alpha]$.

\end{corollary}

\begin{proof}[Proof sketch]

$L = \{w \in \Sigma^*~|~w \otimes \alpha \in \mc{L}(\mc{A})\}$ for some $\omega$-automaton $\mc{A}$. This automaton checks the belonging of suffixes of $\alpha$ to a finite number of $\omega$-regular languages $L_1 \dots L_n$ (as evoked in the beginning of this subsection). We take $f := \max f_1 \dots f_n$ where each $f_i$ is the function given by Corollary \ref{cor:positions} for $L_i$. The reader will get convinced that an automaton $\mc{A}'$ (for finite words) can be built so that $ \{w\lozenge^{f(|w|)}~|~w\in L\} = \{v~|~v\otimes \alpha[:|v|] \in \mc{L}(\mc{A}')\}$.

\end{proof}

The ideas developed above can be applied to obtain the result we claimed.

\begin{corollary}

\label{cor:autaut}

For every advice $\alpha$, $\aut[\alpha] = \aut^\infty[\alpha]$.

\end{corollary}

\begin{proof}[Proof sketch] We follow the same sketch as for Corollary \ref{cor:positions} and extend the $\reg^\infty[\alpha]$-presentation by adding a well-chosen finite number of padding symbols $\lozenge$ behind each word. The function $f$ is now a maximum over the properties of suffixes checked by all the $\omega$-automata for the languages of the presentation and the automata for their complements, since we also need to know when a property does not hold.

\end{proof}

\subsection{Digression: relativization of $\MSO$-formulas}

The results of this subsection can be considered as a digression since they will not be helpful for the rest of our study. As an application of Theorem \ref{theo:infini}, we provide an original normal form for $\MSO$-formulas with free variables when interpreted in a fixed word model. We first take some abbreviations for $\MSO[<]$-formulas:  $\forall x\le y~\phi$ stands for $\forall x~(x\le y \rightarrow \phi)$ and $\exists x\le y~\phi$ for $\exists x~(x\le y \land \phi)$; similarily with set quantifications: $\forall X\le y~\phi$ for $\forall X~((\forall x~(x \in X \rightarrow x \le y)) \rightarrow \phi)$ and $\exists X\le y~\phi$ for $\exists X~((\forall x~(x \in X \rightarrow x \le  y)) \land \phi)$.

\begin{definition}[relativized formulas]

\label{def:restriction}

An $\MSO[<,\Gamma]$ formula $\phi$ with free variables $\overline{X}, \overline{x}, y$ is said to be \emph{relativized under} $y$ if 

$$\phi(\overline{X}, \overline{x}, y) = \bigwedge_{x \in \overline{x}} x \le y \wedge \bigwedge_{X \in \overline{X}} X \le y \wedge \psi(\overline{X}, \overline{x})$$

and every quantification in $\psi$ is of the form $Q z\le y$ or $Q Z\le y$.
\end{definition}

We note that relativized sentences provide a suitable logical formalism to describe the transformations performed by Mealy machines. The proof of the next fact follows from standard logic-automata transformations. 

\begin{fact} 

\label{fact:MMmso}

$\alpha \in \Gamma^\omega$ is the image of $\beta \in \Delta^\omega$ under some Mealy machine if and only there exists a tuple of  $\MSO[<, \Delta]$-formulas $(\phi_a(x))_{a \in \Gamma}$ relativized under $x$, such that for all $n \ge 0$, $\alpha[n] = a$ if and only if $\beta \models \phi_a(n)$. We call such a tuple a \emph{relativized $\MSO$-relabelling}.

\end{fact}

We now consider formulas of the form $\exists y~\phi$ where $\phi$ is relativized under $y$. These formulas are far less expressive than full $\MSO$, since there is always a ``finite proof'' of their validity.

\begin{example}

Let $\phi:= \forall x\exists y~y>x\land P_a(y)$ meaning ``there are infinitely many letters $a$''. There is no relativized sentence equivalent to $\phi$, but among the suffixes $\alpha[n:]$ of a fixed word $\alpha$, this property either always or never holds.

\end{example}

We now show that such formulas (with free variables) are enough to describe the full power of $\MSO$ in a \emph{fixed} infinite word model.

\begin{corollary}

\label{cor:interres}

Let $\phi(\overline{X},\overline{x})$ be a $\MSO[<, \Gamma]$-formula and $\alpha \in \Gamma^\omega$ fixed. There is a formula $\psi(\overline{X}, \overline{x},y)$ relativized under $y$ such that for every tuple $\overline{A}$ of finite sets, and tuple $\overline{a}$ of positions: $\alpha \models \phi(\overline{A},\overline{a})$ if and only if $ \alpha \models \exists y~\psi(\overline{A},\overline{a},y).$

\end{corollary}

\begin{proof}[Proof sktech]

We treat the case of formulas $\phi(X)$ with one free set variable. If $A \subseteq \mb{N}$ is a finite set, denote by $\chi_A\in \{0,1\}^*$ the finite word of length $\max A +1$ with $\chi[n] = 1$ iff $n \in A$. It follows from standard logic-automata translations that $\{\chi_A \otimes \alpha~|~A \text{ finite and }\alpha \models \phi(A)\}$ is an $\omega$-regular language, thus $\{\chi_A~|~A \text{ finite and }\alpha \models \phi(A)\} \in \regi[\alpha]$. From Theorem $\ref{prop:langreg}$ we get that $L := \{\chi_A\lozenge^{f(|\chi_A|)}~|~A \text{ finite and }\alpha \models \phi(A)\} \in \reg[\alpha]$ for some $f: \mb{N} \rightarrow \mb{N}$. Hence there is a finite words automaton $\mc{A}$ such that $L = \{w ~|~w \otimes \alpha[:|w|] \in \mc{L}(\mc{A})\}$. It can be translated back into a formula $\exists y~ X \le y \wedge \psi(X,y)$ with restricted quantifications, where $X$ describes the possible set of positions labelled by $1$.

\end{proof}

\section{Complexity of advices when describing structures}

\label{sec:msot}

After the first results of the previous section on advice automatic structures, we are now able to understand which preorder they describe over infinite words. Corollary \ref{cor:autaut} implies in particular that $\aut[\alpha] \subseteq \aut[\beta]$ if and only if $\aut^\infty[\alpha] \subseteq \aut^\infty[\beta]$. The objective of this section is to show equivalence with $\waut[\alpha] \subseteq \waut[\beta]$ and give several other characterizations. The climax lies in Theorem \ref{theo:caraut} and Theorem \ref{theo:2WFT}, where we relate our notions to well-known logical transformations and finite transducers.

  \begin{definition}

A ($k$-copying) \emph{$\MSO$-transduction ($\MSOT$)} from $\Delta^\omega$ to $\Gamma^\omega$ is a tuple of $\MSO[<, \Delta]$-formulas with free first-order variables.

\vspace*{-0.3cm}

$$(\phi^a_1(x))_{a \in \Gamma} \dots (\phi^a_k(x))_{a \in \Gamma}, (\phi^{<}_{i,j}(x,y))_{1 \le i,j \le k})$$

\end{definition}

The semantics of an $\MSOT$ $\tau$ is defined as that of an $\MSO$-interpretation in $k$ disjoint copies of a host word structure. More precisely, the structure $I_\tau(\mf{W}^\beta)$ (not necessarily a word) has signature $\{<, (P_a)_{a \in \Gamma}\}$ and is defined as follows:

\begin{itemize}

\item $ \dom(I_\tau(\mf{W}^\beta)) = \bigcup_{\substack{1 \le i \le k}} \{(n,i)~|~\text{ there is } a \in A \text{ such that }\beta \models \phi^a_i(n)\}$;

\item if $(n,i) \in \dom(I_\tau(\mf{W}^\beta))$, then $(n,i) \in P_a$ if and only if $\beta \models \phi^a_i(n)$;

\item if $(m,j) \in \dom(I_\tau(\mf{W}^\beta))$, then $(n,i) < (m,j)$ if and only if $\mc{U}\models \phi^{<}_{i,j}(n,m)$.
\end{itemize}

Since we are interested in transformations between words, we only consider the case when $I_\tau(\mf{W}^\beta)$ is a word structure (what is syntactically definable by adding an \mbox{$\MSO[<,\Delta]$}-sentence for the domain). Each $\MSO$-transduction $\tau$ then realizes a (partial) function $\tau: \Delta^\omega \rightarrow \Gamma^\omega$ whose domain is $\{\beta \in \Delta^\omega~|~I_{\tau}(\mf{W}^\beta) \text{ is (isomorphic to) a word structure}\}$, the image $\tau(\beta)$ of $\beta$ being the unique $\alpha$ such that $I_\tau(\mf{W}^\beta) \simeq \mf{W}^\alpha$.

The reader is asked to keep in mind that $\MSOT$ define a certain class of functions on infinite strings, even if our main concern is only the existence of a transduction between two fixed words. We write $\alpha \pc_{\MSOT} \beta$ if there is a $\MSO$-transduction $\tau$ such that $\tau(\beta) = \alpha$.

\begin{remark}
$\MSO$-relabelings (see Definition \ref{def:relabel}), relativized $\MSO$-relabelings, and $1$-dimensional $\MSO$-interpretations can all seen as syntactical fragments of 1-copying $\MSOT$.

\end{remark}

\begin{remark}

Even if $\MSO$-interpretations in general are not closed under composition, it is the case of $\MSOT$ (see e.g. \cite{alur2012}, the problem of tuples of sets disappears). Thus $\pc_{\MSOT}$ is transitive, and is even a preorder over infinite words. Furthermore, the composition of an $\MSOT$ and a $\S$-interpretation can be realized by an unique $\S$-interpretation.

\end{remark}

\begin{example} \item
\label{ex:msotrob}
\begin{enumerate}

\item If $\alpha \REGi \beta$ then $\alpha \pc_{\MSOT} \beta$ (thus $\pc_{\MSOT}$ is a more generic notion of comparison than the preorders of Section \ref{sec:reglang}, we shall see that the increase of power is strict);

\item modifying a finite part of $\alpha$ does not change its $\MSOT$-degree;

\item if the $\mu_n$ are the morphisms of Definition \ref{def:duplication}, then $ \mu_n(\alpha) \pc_{\MSOT} \alpha$ and $\alpha \pc_{\MSOT} \mu_n(\alpha)$ for all $n \ge 1$;

\item if $w$ is a finite word, we denote by $\widetilde{w}$ its mirror image; if $\alpha:= w_1 \# w_2 \# \cdots \in (\Gamma^*\#)^\omega$, let $\widetilde{\alpha}:= \widetilde{w_1} \# \widetilde{w_2} \#\cdots$; then $\widetilde{\alpha} \pc_{\MSOT} \alpha$.

 \end{enumerate}

\end{example}

 \begin{figure}[h!]
    \begin{center}
      \begin{tikzpicture}[scale=0.7]
        \tikzstyle{etat} = [draw, fill=white];
         \node () at (-1.5,0) {$\alpha:=$};
         \node (1) at (0,0) {$a$};
         \node  (2) at (1.5,0) {$b$};
         \node  (3) at (3,0) {$\#$};
         \node  (4) at (4.5,0) {$b$};
         \node  (5) at (6,0) {$a$};
         \node  (6) at (7.5,0) {$a$};
         \node  (7) at (9,0) {$\#$};
         \node  (8) at (10.5,0) {$\dots$};
         
         \draw[->] (1) edge[]  node [above] {} (2);
         \draw[->] (2) edge[]  node [above] {} (3);
         \draw[->] (3) edge[]  node [above] {} (4);
         \draw[->] (4) edge[]  node [above] {} (5);
         \draw[->] (5) edge[]  node [above] {} (6);
         \draw[->] (6) edge[]  node [above] {} (7);
         \draw[->] (7) edge[]  node [above] {} (8);
         
         \node () at (-1.5,-1.5) {$\widetilde{\alpha}=$};
         \node (1) at (0,-1.5) {$a$};
         \node  (2) at (1.5,-1.5) {$b$};
         \node  (3) at (3,-1.5) {$\#$};
         \node  (4) at (4.5,-1.5) {$b$};
         \node (5) at (6,-1.5) {$a$};
         \node  (6) at (7.5,-1.5) {$a$};
         \node  (7) at (9,-1.5) {$\#$};
         \node  (8) at (10.5,-1.5) {$\dots$}; 
         
         \draw[->] (2) edge[]  node [above] {} (1);
         \draw[->] (1) edge[bend right = 30]  node [above] {} (3);
         \draw[->] (3) edge[bend right = -30]  node [above] {} (6);
         \draw[->] (6) edge[]  node [above] {} (5);
         \draw[->] (5) edge[]  node [above] {} (4);
         \draw[->] (4) edge[bend right = 30]  node [above] {} (7);       
         \draw[->] (7) edge[bend right = -30]  node [above] {} (8);      
      \end{tikzpicture}
    \end{center}
    \vspace*{-0.7cm}
    \caption{\small \label{fig:MSOT} Reversing the factors with an $\MSOT$}
\end{figure}
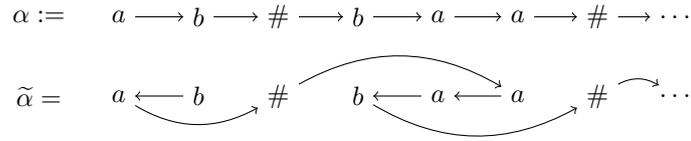

\subsection{From automatic structures to $\MSO$-transductions}

When searching a complete structure of an advice, a naive idea is that $\mf{W}^\alpha \in \auti[\beta]$ if and only if $\auti[\alpha] \subseteq \auti[\beta]$. However, this statement will turn out to be false. We need a stronger object that is presented in Definition \ref{def:pf}.

\begin{definition}[\cite{loding2007}] \label{def:pf}

Let $\mf{A} = \langle A, R_1 \dots R_n \rangle$ be a structure, we define its \emph{weak powerset structure} $\mc{P}^f(\mf{A})$ as the structure $\langle \mc{P}^f(A), R'_1 \dots R'_n, \subseteq \rangle$ where:
\item
\begin{itemize}

\item $\mc{P}^f(A)$ is the weak powerset (set of finite subsets) of $A$;

\item $\subseteq$ is the inclusion relation on $\mc{P}^f(A)$;

\item $R'_i(A_1, \dots, A_{r_i})$ holds in $\mc{P}^f(\mf{A})$ if and only if $A_1, \dots A_{r_i}$ are singletons $\{a_1\}, \dots, \{a_{r_i}\}$ and $R_i(a_1, \dots a_{r_i})$ holds in $\mf{A}$.

\end{itemize}

\end{definition}

\begin{remark} $\mf{A}$ is $\FS$-interpretable in $\mf{B}$ if and only if $\mf{A}$ is $\FO$-interpretable in $\mc{P}^f(\mf{B})$.
\end{remark}

\begin{fact} \label{fact:pf}  $\auti[\alpha]$ is the class of structures $\FO$-interpretable in $\mc{P}^f(\mf{W}^\alpha)$ (see Proposition \ref{prop:si}). We have $\auti[\alpha] \subseteq \auti[\beta]$ if and only if $\mc{P}^f(\mf{W}^\alpha) \in \auti[\beta]$.

\end{fact}

This result provides a characterization which is abstract and, in some respects, trivial. Nevertheless, we get the intuition that powerset structures are a key notion to understand advice automaticity. In the sequel, a ($\Delta$-labelled) tree structure has the form $\langle A, <, (P_a)_{a \in \Delta} \rangle$ where the domain $A$ is a prefix-closed subset of $\{0,1\}^*$, $w<w'$ holds whenever $w$ is a prefix of $w'$ and the $P_a$ label the nodes of $A$ with $a \in \Delta$. Word structures are particular trees.

\begin{theorem}[\cite{loding2007}, Corollary 4.4] \label{theo:loding} Let $\mf{A}$ a structure and $\mf{T}$ a tree structure. If $\mc{P}^f(\mf{A})$ is $1$-dimensionally injectively $\FS$-interpretable in $\mf{T}$, then $\mf{A}$ is $1$-dimensionally injectively $\WMSO$-interpretable in $\mf{T}$.

\end{theorem}

In the case of advice automatic structures, Theorem \ref{theo:loding} is at the same time too generic and too restrictive. On the one hand, we only use interpretations in word structures $\mf{W}^\alpha$. On the other hand, we need arbitrarily dimensional $\FS$-interpretations, and they are not supposed to be injective. We will manage to meet this conditions, up to a slight modification of the advice, and the $\WMSO$-interpretation will be transformed into a more generic $\MSOT$.

\begin{corollary}  \label{cor:autincl}
	If $\mc{P}^{f}(\mf{W}^\alpha) \in \auti[\beta]$, then $\alpha \pc_{\MSOT} \beta$.
\end{corollary}

\begin{proof} 

Assume the hypothesis holds. Then by Propositions \ref{prop:injective} and \ref{prop:binary}, $\mc{P}^{f}(\mf{W}^\alpha)$ has an injective $\regi[\mu_n(\beta)]$-presentation over a binary alphabet, for some $n \ge 1$. Thus  $\mc{P}^{f}(\mf{W}^\alpha)$ is $1$-dimensionally injectively $\FS$-interpretable in $\mf{W}^{\mu_n(\beta)}$ (remarks above). By applying Theorem \ref{theo:loding}, $\mf{W}^\alpha$ is $1$-dimensionally injectively $\WMSO$-interpretable in $\mf{W}^{\mu_n(\beta)}$. Such interpretations are a particular case of $\MSOT$, so $\alpha \pc_{\MSOT} \mu_n(\beta)$. Since $\mu_n(\beta) \pc_{\MSOT} \beta$ (Example \ref{ex:msotrob}), composing both transductions provides $\alpha \pc_{\MSOT} \beta$.

\end{proof}

Since the proof Theorem \ref{theo:loding} in \cite{loding2007} is rather long and involved, we provide in Appendix \ref{proof:theo:caraut} a direct and self-contained proof of Corollary \ref{cor:autincl}. It avoids useless work in the specific case of infinite words and arbitrarily dimensional interpretations.

\begin{remark} Corollary \ref{cor:autincl} can be extended to presentations using tree languages with (infinite) tree advice (see e.g. \cite{abu2017} for a definition).

\end{remark}

\begin{remark}[uniformly automatic classes]

Let $P$ a set of infinite words. A class of structures $\mc{C}$ (over a given signature) is said \emph{uniformly automatic with advice set $P$} if there are fixed automata whose languages with advice $\alpha$ describe presentations of each structure in $\mc{C}$ when $\alpha$ ranges in $P$ \cite{abu2017}. In particular, if $P$ is $\omega$-regular, the $\FO$-theory of the class $\mc{C}$ is decidable. Since the proof of Theorem \ref{theo:loding} only depend of the automata for the presentation of $\mc{P}^f(\mf{A})$, it can be generalized to show that if the uniform classes with $P \subseteq \Gamma^\omega$ are also uniform with $Q \subseteq \Delta^\omega$, then there is an $\MSO$-transduction $\tau$ such that $\tau(Q) = P$.
\end{remark}

We now have all the ingredients to provide effortlessly a useful and elegant characterization for the inclusion of classes.

\begin{theorem} \label{theo:caraut}

The following conditions are equivalent:

\item

\begin{enumerate}

\item $\waut[\alpha]\subseteq \waut[\beta]$; 

\item $\aut^\infty[\alpha]\subseteq \aut^\infty[\beta]$;

\item $\alpha \pc_{\MSOT} \beta$.

\end{enumerate}

\end{theorem}

\begin{proof}[Proof sktech.] We use Proposition \ref{prop:si} several times. The way from $1.$ to $2.$ is a consequence of Theorem \ref{theo:countable}. If $2.$ holds, we show that $\mc{P}^f(\mf{W}^\alpha)$ has an injective binary $\regi[\beta']$-presentation for some infinite word $\beta'$ so that $\beta' \pc_{\MSOT} \beta$. As remarked above, $\mc{P}^f(\mf{W}^\alpha)$ is thus $1$-dimensionally injectively $\FS$-interpretable in the tree $\mf{W}^{\beta'}$, hence Theorem \ref{theo:loding} provides a $1$-dimensionally $\WMSO$-interpretation of $\mf{W}^\alpha$ in $\mf{W}^{\beta'}$, what implies $\alpha \pc_{\MSOT} \beta'$. Composing $\MSOT$ concludes that $\alpha \pc_{\MSOT} \beta$. If $3.$ is true and $\mf{A}$ is $\S$-interpretable in $\mf{W}^\alpha$, then $\mf{A}$ is $\S$-interpretable in $\mf{W}^\beta$ by some composition argument.
\end{proof}

As a consequence, all the preorders defined by advice-presentable structures converge towards the same comparison via $\MSO$-transductions. This point gives a deep theoretical meaning to their study. Another virtue of Theorem \ref{theo:caraut} is the ability to translate immediately the results of Example \ref{ex:msotrob} in terms of advice automatic structures.

\begin{example} \item
\label{ex:autrob}
\begin{enumerate}

\item If $\alpha \REGi \beta$ then $\aut[\alpha] \subseteq \aut[\beta]$;
\item modifying a finite part of $\alpha$ does not modify $\aut[\alpha]$;
\item  if $\alpha \in (\Gamma^*\#)^\omega$, then $\aut[\alpha] = \aut[\widetilde{\alpha}]$.
 
 \end{enumerate}

\end{example}

\subsection{An equivalent computational model: two-way transducers}
We will complete our parallel with transductions via an equivalent simple machine model. Furthermore, it will be very useful to describe the structural properties of the preorder.

\begin{definition}

A \emph{two-way finite transducer} ($\WFT$) is a 6-tuple $( Q,q_0,  \Delta \uplus \{\vdash\}, \Gamma, \delta, \theta)$ where $Q$ is the finite set of states, $q_0 \in Q$ is initial, $\Delta$ is the input alphabet, $\Gamma$ is the output alphabet, $\delta: Q \times (\Delta\uplus \{\vdash\}) \rightarrow Q \times \{\triangleleft, \triangleright\}$ is the (partial) transition function, and $\theta: Q \times (\Delta\uplus \{\vdash\})\rightarrow \Gamma^*$ is the (partial) output function.

\end{definition}

A $\WFT$ has a two-way read-only input tape and a one-way output mechanism. The component $\{\triangleleft,\triangleright\}$ determines the left or right move of the head on the input tape. When the $\WFT$ is given $\beta \in \Delta^\omega$ as an input word, this tape contains $\vdash \beta$ (adding a symbol $\vdash$ helps the transducer to notice the beginning of its input when going left). The definition of the (partial) function $\Delta^\omega \rightarrow \Gamma^\omega$ realized the $\WFT$ follows like for Mealy machines.

\begin{remark} The transducer is said to be \emph{one-way} ($\DFT$, or just finite transducer) if all its transitions are of the form $(q, \triangleright)$. Mealy machines are a particular case of $\DFT$.
\end{remark}

\begin{example}

There is a three-state $\WFT$ outputting $\widetilde{\alpha}$ on every $\alpha \in (\Gamma \#)^\omega$. Its behavior is the following: scan a maximal $\#$-free block, read it in a reversed way while outputting, then output $\#$ and move to the next block.

\end{example}

Write $\alpha \pc_{\WFT} \beta$ if $\alpha$ is the image of $\beta$ under a function realized by some $\WFT$. We now give some basic properties of these transductions.

\begin{fact} \label{fact:ultper} If $\alpha$ is ultimately periodic, for every string $\beta$ we have $\alpha \pc_{\WFT} \beta$.

\end{fact}

\begin{lemma} \label{lem:wftper} Let $\mc{T}$ be a $\WFT$ transforming $\beta$ into $\alpha$. If $\alpha$ is not ultimately periodic, there is an integer $N$ such that the run of $\mc{T}$ does not visit more than $N$ times each position of the string $\vdash \beta$.
\end{lemma}

\begin{proof} Let $N$ be the number of states of $\mc{T}$. If $\mc{T}$ visits more than $N$ times a position, it is caught in a loop and must output an ultimately periodic word.

\end{proof}

When considering definable \emph{functions} between \emph{finite} strings, a well-known equivalence holds between $\MSOT$ and $\WFT$ (Theorem \ref{theo:engel}). The definitions of $\MSOT$ and $\WFT$ have to be slightly sharpened to get the exact correspondence, see details in \cite{engelfriet2001}.

\begin{theorem}[\cite{engelfriet2001}] \label{theo:engel}

(Partial) functions over finite words $\Delta^*\rightarrow \Gamma^*$ definable by $\MSOT$ are the (partial) functions realized by $\WFT$.
\end{theorem}

Fairly recently, this result was extended to functions between infinite strings, but some complications quickly appear: deciding the validity of $\MSO$-sentences is not always possible without reading the (variable) input entirely. Thus $\WFT$ alone are not powerful enough and they need extra features like $\omega$-regular lookahead, i.e. ability to check instantly $\omega$-regular properties of the suffixes of the input starting in the position of the reading head.

\begin{theorem}[\cite{alur2012}] \label{theo:alur}

(Partial) functions over infinite words $\Delta^\omega\rightarrow \Gamma^\omega$  definable by $\MSOT$ are the (partial) functions realized by $\WFT$ with $\omega$-regular lookahead whose runs always visit the whole input string.

\end{theorem}

When looking closely at Theorem \ref{theo:engel} and Theorem \ref{theo:alur} in the light of our previous results, a question arises naturally: it is possible to get rid of the lookaheads when fixing the input infinite word? Indeed, we have always considered transformations from a \emph{fixed} word and we noticed in Subsection \ref{sec:autstr} that this restriction simplified certain notions. Theorem \ref{theo:2WFT} gives a positive answer. This involved result is not a direct consequence of Theorem \ref{theo:alur}, since we are not aware of a simple manner to remove the $\omega$-lookaheads when fixing the input.

\begin{theorem}

\label{theo:2WFT}

$\alpha \pc_{\MSOT} \beta$ if and only if $\alpha \pc_{\WFT} \beta$.

\end{theorem}

\begin{proof}[Proof sketch]  If $\alpha \pc_{\WFT} \beta$, the result follows from Theorem \ref{theo:alur}. Indeed the transformation can be computed by some $\WFT$ (with a trivial $\omega$-lookahead)  whose run visits the whole input. Assume now that $\alpha \pc_{\MSOT} \beta$. It follows from \cite{alur2012} that $\alpha$ can be computed from $\beta$ by an \emph{$\omega$-streaming string transducer ($\SST$)}. We show in Appendix \ref{proof:theo:2WFT} that an $\SST$ can be transformed into a $\WFT$ with a \emph{lookbehind} feature, when the input word is fixed. Lastly, the lookbehind can be removed by some standard techniques (a lookbehind only deals with a \emph{finite} part of the input, which is not the case of an $\omega$-lookahead).

\end{proof}

\begin{remark}

Without Theorem \ref{theo:2WFT}, it is not clear that $\pc_{\WFT}$ is transitive.

\end{remark}

We finally note that the definition of $\WFT$ can slightly simplified.

\begin{fact}
We can consider w.l.o.g. in Theorem \ref{theo:2WFT} that the transducer reads and moves on an input tape containing $\beta$ instead of $\vdash \beta$.
\end{fact}

\begin{proof} If $\alpha$ is ultimately periodic, the result is obvious by Fact \ref{fact:ultper}. Else, according to Lemma \ref{lem:wftper}, there is $N \ge 0$ such that the transducer does not visit position $0$ of $\vdash \beta$ more than $N$ times. Thus after a certain time $n_0$ the run never visits $\vdash$. What it output before $n_0$ can be hardcoded, and the rest of the computation can be done on $\beta$ directly.

\end{proof}

\section{The two-way transductions hierarchy}

\label{sec:hierarchy}

We initiate in this section a study of the previous two-way transductions between infinite words. It can equivalently be seen as the preorder defined by $\MSOT$, or classes $\aut[\alpha]$, $\auti[\alpha]$ and $\waut[\alpha]$; but the $\WFT$ formulation is - as predicted above - the easiest way to deduce interesting statements. We shall use the term \emph{$\WFT$ hierarchy} to describe the ordered set of $\WFT$-degrees (i.e. equivalence classes of $\pc_{\WFT} \cap \succcurlyeq_{\WFT}$).

A more or less similar work has been done in \cite{endrullis2015}, with the relation $\pc_{\DFT}$ defined by computability via $\DFT$. This definition clearly describes a preorder. Even if no previous research exists on the $\WFT$ hierarchy, we shall see that several results on the $\DFT$ can be adapted in our context, after a variable amount of work. Note that $\pc_{\DFT} \subseteq \pc_{\WFT}$.

\begin{proposition} \label{prop:WFThierarchy} \item 

\begin{enumerate}

\item There are uncountably many distinct $\WFT$-degrees;

\item a set of $\WFT$-degree has an upper bound if and only if it is countable;

\item the $\WFT$ hierarchy has no greatest degree;

\item every $\WFT$-degree contains a binary string.

\end{enumerate}

\begin{proof}[Proof sketch] The proofs of similar statements for $\DFT$ in \cite{endrullis2015} raise no specific issue and their adaptation is straightforward.

\end{proof}

\end{proposition}

\begin{remark} \label{rem:MSOvsREL} Considering binary strings is thus sufficient to describe all the degrees. Comparing this result with Proposition \ref{prop:subword} shows that the preorder $\pc_{\regi}$ defined by $\MSO$-relabelings is strictly weaker than $\pc_{\WFT} = \pc_{\MSOT}$.

\end{remark}

As a consequence of Proposition \ref{prop:WFThierarchy}, the $\WFT$ hierarchy is not trivial. We now show that it is fine-grained enough to distinguish ultimately periodic words.

\begin{proposition}

\label{prop:ult}

Ultimately periodic words are the least $\WFT$-degree.

\end{proposition}

\begin{proof} Assume that $\alpha$ is ultimately periodic. Fact \ref{fact:ultper} concludes $\alpha \pc_{\MSOT} \beta$ for all $\beta$. Conversely, if $\gamma \pc_{\WFT} \alpha$, we show that $\gamma$ is ultimately periodic. Since $\alpha \pc_{\WFT} \square^\omega$ (take $\beta = \square^\omega$ in the previous argument), we get $\gamma  \pc_{\WFT} \square^\omega$. Let $\mc{T}$ be the $\WFT$ computing this transformation and $(q_i,n_i)_{i \ge 0}$ its run on $\square^\omega$ (sequence of tuples state/position). There exists $j \ge 0$ and $k \ge 1$ such that $q_{j} = q_{j+k}$. Due to determinism and invariance of $\square^\omega$ by translation, this sequence of moves must be repeated (note that necessarily $n_{j+k} \ge n_j$), what shows the ultimate periodicity of $\gamma$.
\end{proof}

\begin{remark} \label{rem:ultdft} Ultimately periodic words are also the least $\DFT$ degree.
\end{remark}

This easy result shows, through the equivalences of Section \ref{sec:msot}, that non-trivial advices \emph{strictly} increase the class of presentable structures (it was not obvious before). We equivalently provided a characterization for automaticity of certain structures.

\begin{corollary} \label{cor:autspf} $\mc{P}^f(\mf{W}^\alpha)$ is automatic if and only if $\alpha$ is ultimately periodic.

\end{corollary}

\begin{remark} \label{rem:autw}Some structures $\mf{W}^{\alpha}$ are automatic for $\alpha$ non-ultimately periodic \cite{barany2008}, hence the automaticity of $\mf{W}^{\alpha}$ is not equivalent to that of $\mc{P}^f(\mf{W}^\alpha)$. One of the open questions in this field is to understand when $\mf{W}^\alpha$ is automatic, and Corollary \ref{cor:autspf} may be considered as a small step in this direction.

\end{remark}

We now turn to a more involved statement. A sequence $\beta$ is said to be \emph{prime} if it is a minimal but non-trivial word. Formally, $\beta$ non ultimately periodic is prime in the $\WFT$ hierarchy if for all $\alpha \pc_{\WFT} \beta$, either $\beta \pc_{\WFT} \alpha$ or $\alpha$ is ultimately periodic. The existence of prime sequences shows in particular that the $\WFT$ hierarchy is not dense.

\begin{theorem}

\label{theo:prime}

The sequence $\pi := \prod_{n=0}^\infty 0^n1$ is prime in the $\WFT$-hierarchy.

\end{theorem}

\begin{proof}

We show in Appendix \ref{proof:theo:prime} that if $\alpha \pc_{\WFT} \pi$ then $\alpha \pc_{\DFT} \pi$. Since $\pi$ is prime in the $\DFT$ hierarchy \cite{endrullis2015}, either $\pi \pc_{\DFT}  \alpha$ or $\alpha$ is in the least $\DFT$-degree, which is also the set of ultimately periodic words.

\end{proof}

Classifying all infinite strings may neither be relevant nor useful in practice. We now look at two particular classes of infinite words closed under $\WFT$ transformations.

\begin{proposition}[subhierarchies] \item

\begin{enumerate}

\item If $\alpha \pc_{\WFT} \beta$ and if $\beta$ is computable, then $\alpha$ is computable;

\item if $\alpha \pc_{\WFT} \beta$ and if $\mf{W}^\beta$ has a decidable $\MSO$-theory, then $\mf{W}^\alpha$ has a decidable $\MSO$-theory.

\end{enumerate}

\end{proposition}

\begin{proof} Computability is immediate and decidability follows from the equivalence with $\MSOT$.

\end{proof}

\begin{fact} The string $\pi$ has a decidable $\MSO$-theory (see e.g. \cite{barany2008}).

\end{fact}

\begin{proposition} There exists a greatest degree of computable strings in the $\WFT$ hierarchy.
\end{proposition}

\begin{proof} According to \cite{endrullis2015}, there exists a computable string $\tau$ such that for any computable string $\alpha$, $\alpha \pc_{\DFT} \tau$, thus \emph{a fortiori} $\alpha \pc_{\WFT} \tau$.

\end{proof}

\begin{fact}

The $\MSO$ theory of $\tau$ is not decidable, since there exists computable strings with undecidable $\MSO$-theory.

\end{fact}

Figure \ref{fig:hierarchy} summarizes the previous results. Note that the $\WFT$-degree of ultimately periodic sequences, the  $\WFT$-degree of $\pi$ and the  $\WFT$-degree of $\tau$ have to be distinct. Several challenging issues naturally arise about the structure of the $\WFT$ hierarchy and its subhierarchies. Among others, an interesting question is to describe the degrees of well-known sequences with decidable $\MSO$-theory, for instance morphic words \cite{barany2008}.

\begin{figure}

\begin{center}
       \begin{tikzpicture}[scale=0.6]

      \draw[] (0,-1) -- (6.2,5.5) -- (-6.2,5.5) -- (0,-1);
      \draw[] (0,-1) -- (10,7.5) -- (-10,7.5) -- (0,-1);
      \node at (-0.5,2) {{\scriptsize no degree}};
      \node at (-0.3,1.6) {{\scriptsize between}};
       \node at (0,1.2) {{\scriptsize $\pi$ and $0$}};
      
       \node at (2.8,0) {{\scriptsize Least degree:}};
        \node at (2.9,-0.4) {{\scriptsize ultimately periodic}};

        \node[draw,circle] (0) at (0,-0.2) {$0$};
        
         \node[] at (0.2,5) {\textbf{Strings with decidable $\MSO$-theory}};
         
         \node[] at (5.5,7) {\textbf{Computable strings}};
        
        \node[draw,circle] (2) at (-2.6,3) {$\pi$};
        
         \node at (-1.1,3.2) {{\scriptsize A prime}};
         \node at (-1.3,2.8) {{\scriptsize degree}};
        
        \node[draw,circle] (1) at (0,6.6) {$\tau$};
        
        \node at (-2.7,6.8) {{\scriptsize Greatest degree of}};
        \node at (-2.8,6.4) {{\scriptsize computable strings}};
        
           \draw[->] (2) edge (0);

       \end{tikzpicture}

       \caption{\label{fig:hierarchy} An partial look on the $\WFT$ hierarchy}
       
              \end{center}

\end{figure}

\section{Conclusion and outlook}

\textbf{\textsf{Preorders of advices, logic and transducers.}} Our first concern in this paper was the study of various preorders over infinite words, related to the notion of advice strings. The results draw a generic correspondance between definability with advice, logical transductions and machine transductions. Table \ref{fig:sum} summarizes this philosophy in an elegant way, note that the notion of (relativized) $\MSO$-relabelings is less standard than $\MSOT$. The gap between $\MSO$-relabelings and $\MSOT$ shows that having basic knowledge on the languages is far from being sufficient to understand the richness of presentable structures.

\begin{table}[h!]
\centering
\begin{tabular}{| l | c | c | c |}

\hline

&$\reg$&& $\aut$\\

\textbf{Advice}&& $\reg^\infty$ & $\aut^\infty$ \\

&&$\wreg$& $\waut$\\

\hline

\textbf{Logic} & rel. $\MSO$-relabelings & $\MSO$-relabelings & $\MSOT$\\

\hline

\textbf{Machine}&Mealy machines & $\omega$-regular functions & $\WFT$ \\

\hline

\end{tabular}

\caption{\label{fig:sum}\small Equivalent definitions for preorders over $\omega$-words}
\end{table}

\textbf{\textsf{A meaningful hierarchy of infinite words.}} Two-way transductions appear here to be more basic than relations defined by one-way machines, since they are clearly motivated by logical issues. Furthermore, it fits our informal conditions to be a ``good'' complexity measure over infinite words. A more involved study of the $\WFT$ hierarchy may help classifying certain hierarchies of structures, or even understand standard automatic presentations. We recall that such transductions over infinite words are (rather) unexplored.


\bibliographystyle{alpha}
\bibliography{ms}

\newpage


\tableofcontents

\appendix

\pagebreak

\section{Proof of Proposition \ref{prop:carwreg}}

\label{proof:prop:carwreg}

We first give a logical description for advice unary languages.

\begin{lemma} \label{lem:unary}

Let $\beta \in \Delta^\omega$. A language $U \subseteq 0^+$ is in $\reg^\infty[\beta]$ if (and only if) there is an $\MSO[<, \Delta]$-formula $\phi(x)$ such that $\beta \models \phi(n) \iff 0^{n+1} \in U $.

\end{lemma}

\begin{proof}

Assume $U \in \reg^\infty[\beta]$. There is an $\omega$-regular language $L \subseteq (\{0, \square\} \times \Delta)^\omega$ such that $U = \{w \in 0^*~|~w\otimes\alpha \in L\}$. Automata-logic translations show that the set $L$ can be described by an $\MSO[<, \{0,\square\}\times \Delta]$-sentence $\psi$. Let $\phi_a(x)$ be the $\MSO[<, \Delta]$-formula obtained from $\psi$ by adding one free first-order variable $x$ and replacing each $P_{(0,b)}(y)$ by $P_{b}(y) \wedge y \le x$ and each $P_{(\square,b)}(y)$ by $P_{b}(y) \wedge y > x$. An induction then provides $\beta \models \phi(n) \iff 0^{n+1} \in U $. The converse results from a similar argument.

\end{proof}

$1 \Rightarrow 4$. Assume that $\reg^\infty[\alpha] \subseteq \reg^\infty[\beta]$, then $\Pref(\alpha) \in \reg^\infty[\beta]$.  Given a letter $a \in \Gamma$, we get $\{0^{n+1}~|\alpha[n]=a\} \in \reg^\infty[\beta]$ using closure properties of Proposition \ref{prop:closewreg}. Applying Lemma \ref{lem:unary} for each letter then builds an $\MSO$-relabelling.

$4 \Rightarrow 2 \Rightarrow 3$. If we have a relabelling $(\phi_a(x))_{a \in \Gamma}$, let $\Phi$ be the formula $\bigwedge_a (P_a(x) \leftrightarrow \phi_a(x))$ translated on the signature $\Gamma \times \Delta$. $\Phi$ describes an $\omega$-regular language $L' \subseteq (\Gamma \times \Delta)^\omega$ such that $\{\alpha\} = \{w~|~w\otimes\beta \in L'\}$. Thus $\{\alpha\} \in \wreg[\beta]$ and $2$ follows from closure properties of $\omega$-regular languages. On the other hand, uniformization theorems (transformations of relations into functions, see \cite{carayol2012}) applied to $L'$ give $3$.

$3\Rightarrow 1$. If $\alpha$ is the image of $\beta$ under some $\omega$-regular function, closure properties of $\omega$-regular languages show that every automaton with advice $\alpha$ can also use $\beta$.

\begin{remark}

The proof also gives equivalence with $\Pref(\alpha) \in \reg^\infty[\beta]$ and $\{\alpha\} \in \wreg[\beta]$. These languages are, in a certain sense, complete for the advice classes.

\end{remark}

\section{Proof of Theorem \ref{theo:infini}}

\label{proof:lem:restrict}

We split the proof in two independent parts to treat either $\FO$-definable languages only, or $\omega$-regular in a general way. The first case uses logical notions, whereas the second is highly based on structural properties of ($\omega$-)automata.

\subsection{$\FO$-definable languages}

At first glance, there is no reason why we could only check a finite part of a suffix to deduce an $\FO$-definable $\omega$-property. Indeed, it seems hard to formalise the intuition that, since our models all are suffixes of a given word, what happens infinitely often is always true, because first-order sentences are somehow too complex. However, there is no need to despair, because $\FO$ has the same expressive power (over infinite words) as \emph{linear temporal logic} ($\LTL$) - a result known as Kamp's theorem. In order to avoid possible confusions since several equivalent syntaxes exist, we now recall a syntax of $\LTL$.

\begin{definition}[{$\LTL[\Gamma]$}]

The set of $\LTL[\Gamma]$-formulas is defined inductively: every $a \in \Gamma$ is a formula, and $\top$ as well (atoms), if $\phi_1$ and $\phi_2$ are formulas, so are $\phi_1 \land \phi_2$, $\phi_1 \lor \phi_2$,$\neg \phi_1$, $\X \phi_1$ and $\phi_1 \U \phi_2$.

\end{definition}

We may use the following abbreviations: $\bot:= \neg \top$ and $\G \phi:= \neg (\top \U \neg \phi)$. The semantics of $\LTL$ in $\omega$-words being well-known and quite intuitive, we do not recall it here.

\begin{theorem}[Kamp, future-only fragment \cite{rabinovich2014}]

\label{theo:kamp}

Over the word structures of $\Gamma^\omega$, for every $\FO[<, \Gamma]$-sentence $\phi$, there is an equivalent $\phi' \in \LTL[\Gamma]$.

\end{theorem}

The next step is to get a negation normal form.

\begin{lemma}

\label{lem:nnfltl}

For every $\LTL$-formula, there is an equivalent formula where negations only appear in front of atoms, and other connectives are $\U$, $\G$, $\X$, $\land$ and $\lor$.

\end{lemma}

\begin{proof}[Proof idea (folklore).]

Apply inductively standard $\LTL$ equivalences, namely $\neg \X \phi \equiv \X \neg \phi$ and $\neg (\phi \U \psi) \equiv (\G \neg \psi) \lor (\neg \psi \U (\neg \psi \land \neg \phi))$.

\end{proof}

We fix an $\FO[<, \Gamma]$-sentence $\phi$ and an infinite word $\alpha \in \Gamma^\omega$. By Theorem \ref{theo:kamp} and Lemma \ref{lem:nnfltl}, there exists a $\LTL[\Gamma]$-formula $\phi'$ in negation normal form and equivalent to $\phi$ (over all word models). The idea is now to remark that, since the word $\alpha$ is fixed, the connective $\G$ is somehow useless. Let the $\G$-subformulas of $\phi'$ be the largest subformulas whose main connective is $\G$. If there is $n \ge 0$ such that $\alpha[n:] \models \G \nu$, the $\G$-subformula $\G \nu$ is said to be \emph{consistent}, in that case $\alpha[m:] \models \G \nu$ for all $m\ge n$.

Let $\phi''$ be the formula obtained from $\phi'$ by replacing each consistent $\G$-subformula by $\top$ and each non-consistent $\G$-subformula by $\bot$. $\phi''$ is in negation normal form and has no longer $\G$ connectives. Furthermore, there is $N \ge 0$ such that for all $m  \ge N$, $\alpha[m:]\models \phi''$ if and only if  $\alpha[m:]\models \phi'$ (take $N$ to be the maximum of all the $n$ from the previous paragraph).

Now we translate back the formula $\phi''$ into an $\FO$-sentence on \emph{finite words}.

\begin{lemma} \label{lem:ltlphi}

For every $\LTL[\Gamma]$-formula $\phi''$ where negations only appear in front of atoms, and other connectives are $\U$,$\X$, $\land$ and $\lor$, there is an $\FO$-sentence $\psi$ such that for all $\beta \in \Gamma^\omega$, $\beta \models \phi''$ if and only if there is $n \ge 0$ such that $\beta[:n] \models \psi$. 

\end{lemma}

\begin{proof}[Proof idea.]

Easy nduction on the $\LTL[\Gamma]$-formula.

\end{proof}

Since Lemma \ref{lem:ltlphi} holds in particular for suffixes $\alpha[m:]$ of $\alpha$, we get that for all $m \ge N$, $\alpha[m:] \models \phi$ if and only if $\exists n \ge m, \alpha[m:n] \models \psi$.

\subsection{General case of $\omega$-regular languages}

We will follow a completely different scheme here. Our approach is based on results of Semenov \cite{semenov1984}, where he gives a characterization of $\omega$-words whose $\MSO$-theory is decidable. This main result is not useful in our context, but one of technical lemmas stated in the proof is of a particular relevance. We briefly recall the definitions of \cite{semenov1984}. A congruence  $\mc{E}$ on $\Gamma^*$ is an equivalence relation of finite index, compatible with concatenation. The key idea of Semenov lies in the following definition.

\begin{definition}

Let $\mc{E}$ be a congruence, an \emph{$\mc{E}$-index} $c$ is a nonempty finite word on the alphabet of equivalence classes, of length at most the index of the congruence. Given a $\mc{E}$-index $c$, define the set of its values $\val(c)\subseteq \Gamma^*$ by induction:

\item

\begin{itemize}

\item if $c$ is only one equivalence class $E$, then $\val(c) = E \cap \Gamma$ (letters);

\item if $c = c'E$ where $E$ is a class, then $w \in \val(c)$ if and only three conditions are met: $w \in E$, every proper suffix/prefix of $w$ belongs to $\val(c')$, every subword of $w$ belonging to $\val(c')$ is either a suffix or a prefix.

\end{itemize}

\end{definition}

It is not hard to see that for every $\mc{E}$-index $c$, $\val(c)$ is a regular set which does not contain a proper subword of its words. The intuition behind this definition is that it helps computing the possible segments of runs in a given automaton.

Denote by $\ind(\mc{E})$ the finite set of all $\mc{E}$-indices. To quantify the occurrences of their values in a given $\alpha \in \Gamma^\omega$, let $\mc{E}^n \alpha:= \{c \in \ind(\mc{E})~|~\exists w \in \val(c) \text{ subword of }\alpha[n:]\}$. It is clear that $\mc{E}^{n+1} \alpha \subseteq \mc{E}^{n} \alpha $, and since there is a finite number of indices, the sequence $(\mc{E}^n\alpha)_{n \ge 1}$ is ultimately constant with value $\mc{E}\alpha:= \bigcap_{n \ge 1} \mc{E}^{n}$. A position $n \ge 0$ in a given $\alpha$ is said to be $\mc{E}$-remote if $\mc{E}^n \alpha = \mc{E} \alpha$ (we reached some kind of stability with respect to the congruence). Note that $\mc{E}$-remote positions of $\alpha$ form a final non-empty segment of $\mb{N}$.

\begin{definition}[regular system]

\label{def:regsys}

Given a congruence $\mc{E}$, a regular $\mc{E}$-system $\mc{R}$ is a mapping from $\mc{P}(\ind(\mc{E}))$ into regular subsets of $\Gamma^*$, such that for every $\omega$-word $\alpha$ and every position $n \ge 0$, there is a unique position $m \ge n$ such that $\alpha[n:m]\in \mc{R}(\mc{E}\alpha)$.

\end{definition}

\begin{example}

If $S$ is a set of $\mc{E}$-indices, let $ L(S) := \bigcap_{c \in S} \Gamma^* \val(c) \Gamma^*$. Let $\mc{R}$ mapping $S$ to the set of words of $ L(S)$ that have no proper prefix in $L(S)$ (smallest elements for prefix-ordering). Then $\mc{R}$ is a $\mc{E}$-regular system.

\end{example}

If $n\ge 0$ is a position, denote by  $\chi_n$ the characteristic sequence $0^n10^\omega$.

\begin{lemma}[\cite{semenov1984}, Lemma 8]

\label{lem:semenov}

For every $\omega$-regular language $L$, there is a congruence $\mc{E}$, a regular $\mc{E}$-system $\mc{R}$ and a finite word automaton $\mc{A}$ such that the following holds. For every $\alpha \in \Gamma^\omega$ and every $\mc{E}$-remote position $n$ (of $\alpha$), $\alpha \in L$ if and only if $\mc{A}$ accepts $\alpha[:m] \otimes \chi_n[:m]$ where $m\ge n$ is the unique position such that $\alpha[n:m] \in \mc{R}(\mc{E}\alpha)$.
\end{lemma}

This result is especially suitable in our context. Remark that given a congruence $\mc{E}$ and an infinite word $\alpha$, there is $N\ge 0$ such that for all $m \ge N$, every position is $\mc{E}$-remote in $\alpha[m:]$  (take $N$ to be any $\mc{E}$-remote position of $\alpha$). This allows us to reformulate a weaker version of Lemma \ref{lem:semenov}.

\begin{lemma}

\label{lem:finitela}

For every $\omega$-regular language $L$ and every infinite word $\alpha$, there exist $N\ge 0$ and a finite word automaton $\mc{D}$ such that for all $m \ge N$:

\item

\begin{itemize}

\item if $\alpha[m:] \in L$ there is a (unique) $i_m\ge m$ such that $\mc{D}$ accepts $\alpha[m:i_m]$;

\item  else there is no $i \ge m$ such that $\mc{D}$ accepts $\alpha[m:i]$.

\end{itemize}
\end{lemma}

\begin{proof} 
We use the notations of Lemma \ref{lem:semenov}. Since $\alpha$ is fixed, $R:= \mc{R}(\mc{E}\alpha)$ is a fixed regular language verifying the subword property evoked in Definition \ref{def:regsys}. Also note that for all $n \ge 0$, $ \mc{R}(\mc{E}\alpha[n:]) = R$ since $\mc{E}\alpha[n:] = \mc{E}\alpha$. Choose $N$ to be such that for all $m \ge N$, every position in $\alpha[m:]$ is $\mc{E}$-remote. In particular $0$ is an $\mc{E}$-remote position of each $\alpha[m:]$, hence we can get rid of the characteristic sequence $\chi_n$ if we take this value. In other words, we build from Lemma \ref{lem:semenov} an automaton $\mc{A}'$ such that for all $m \ge N$, $\alpha[m:]\in L$ if and only if $\mc{A}'$ accepts $\alpha[m:m']$ where $m'\ge m$ is the unique position such that $\alpha[m:m'] \in R$.

Now, the construction of $\mc{D}$ is based on the product of $\mc{A'}$ and an automaton $\mc{A}''$ recognizing $R$. It checks if $\mc{A'}$ accepts when $\mc{A}''$ accepts, what necessarily happens after a finite time.

\end{proof}

Lemma \ref{lem:finitela} provides exactly the elements we needed to achieve the proof.

\section{Self-contained proof of Corollary \ref{cor:autincl}}

\label{proof:theo:caraut}

We show that $\mc{P}^f(\mf{W}^\alpha) \in \aut^\infty[\beta]$ implies $\alpha \pc_{\MSOT} \beta$. We actually adapt and shorten the proof of \cite{loding2007} in our context. We first present a simple class of $\reg^\infty[\beta]$-presentation.

\begin{definition} [homogeneity]
A $\reg^\infty[\beta]$-presentation is said to be \emph{homogeneous} if there is an integer $K\ge 1$ such that the finite words encoding the elements belong to $\bigcup_{1 \le k \le K} k^*$.
\end{definition}

Such presentations are especially simple, in the sense that the only relevant information in an word is its length (and its ``color'' $k$, which is bounded). They are closely related to $\MSO$-transductions, as shown in Lemma \ref{lem:homomso}.

\begin{lemma}

\label{lem:homomso}

If $\mf{W}^\alpha$ has an homogeneous $\reg^\infty[\beta]$-presentation, $\alpha \pc_{\MSOT} \beta$.

\end{lemma}

\begin{proof}[Proof sketch.] Let $L \subseteq \bigcup_{1 \le k \le K} k^*$ be the language encoding the domain of $\mf{W}^\alpha$ in the presentation. We build a $K$-copying $\MSO$-transduction $\tau$ such that $\tau(\beta) = \alpha$. Indeed, according to Lemma \ref{lem:unary}, for $1 \le  k \le K$ fixed, the set of positions $\{n~|~k^n \in L\}$ can be described by a formula $\phi_k(x)$. A similar argument show that there exists a formula $\phi_{<}(x,y)$ to describe the relation $<$.

\end{proof}

\begin{lemma} \label{lem:pro} If $\mc{P}^f(\mf{A}) \in \aut^\infty[\beta]$, then $\mf{A}$ has an homogeneous $\reg^\infty[\beta]$-presentation.

\end{lemma}

Assume Lemma \ref{lem:pro} holds. If $\mc{P}^f(\mf{W}^\alpha) \in \aut^\infty[\beta]$ then $\mf{W}^\alpha$ has  an homogeneous $\reg^\infty[\beta]$-presentation, thus $\alpha \pc_{\MSOT} \beta$ by Lemma \ref{lem:homomso}.

The rest of this section is devoted to the proof of Lemma \ref{lem:pro}.

\subsection*{Proof of Lemma \ref{lem:pro}}

Let $\mf{A}$ be a structure of domain $A$ such that $\mc{P}^f(\mf{A}) \in \aut^\infty[\beta]$. The $\reg^\infty[\beta]$-presentation of $\mc{P}^f(\mf{A}) $ can be assumed injective by Proposition \ref{prop:injective}. Let $\Sigma$ the encoding alphabet and $\nu : P^f(A) \rightarrow \Sigma^*$ the encoding function.

Let $\atoms:= \{w \square^\omega~|~\nu(w) \text{ is a singleton}\}$. Our main purpose is to number the elements of $\atoms$ in a regular-like way. More formally, we build a function $\ind: \atoms \rightarrow \mb{N}$ such that $|\ind^{-1}(n)| \le K$ for all $n \ge 0$ and the language $\{w \otimes 0^n \square^\omega~|~\ind(w) = n\}$ is $\omega$-regular with advice $\beta$ ($K$ being a large enough constant) . Once this is done, the lemma follows almost directly. We denote by $L_{\subseteq}$ the language $\{w \square^\omega \otimes w' \square^\omega~|~\nu(w) \subseteq \nu(w')\}$.

\paragraph*{A bounded-to-one index function}

\label{sec:index}

Thanks to the $\reg^\infty[\beta]$-presentation, there exist two (deterministic Müller) $\omega$-automata $\mc{A}_{\atoms}$ and $\mc{A}_{\subseteq}$ such that $\atoms = \{w~|~w \otimes \beta \in \mc{L}(\mc{A}_{\atoms})\}$ and $L_{\subseteq} = \{w~|~w \otimes \beta \in \mc{L}(\mc{A}_{\subseteq})\}$ (recall that here $w$ itself is a convolution). Let $\qs$ and  $\qi$ be the sets of states of these two automata and let the constant $K_{im}:= (2 |\qi|+1)|\qs|$ ($\ge 2$).

\begin{definition} Let $w \in \atoms$, the set of its \emph{important positions} $\imp(w) \subseteq \mb{N}$ is such that $n \in \imp(w) $ if and only if $|\{w'\in \Sigma^* \square^\omega~|~w[:n]w' \in \atoms\}| > K_{im}$.
\end{definition}

$\imp(w)$ can be seen as the set of prefixes $w[:n]$ whose lecture does not give much information on who is $w$. Note that it is an initial segment of $\mb{N}$ (we can w.l.o.g. assume that it is always non-empty as soon as $A$ is not finite). Furthermore, if $w = v \square^\omega$, the elements $\imp(w)$ are smaller than $|v| +1$. Therefore the following definition makes sense.

\begin{definition} If $w \in \atoms$, $\ind(w)\in \mb{N}$ is $\max \imp(w)$.
\end{definition}

Before giving the essential property of $\ind$, we first remark that  $\{w \otimes 0^n \square^\omega~|~\ind(w) = n\}$ is $\omega$-regular with advice $\beta$ (closure properties).

\label{sec:bound}

\begin{lemma}

\label{lem:bound}

There is a constant $K$ such that for all $n \ge 0$, $|\ind^{-1}(n)| \le K$.

\end{lemma}

The rest of this paragraph is dedicated to the combinatorial proof of this lemma. We first give some notations. Let $\delta_{\atoms}$ be the transition function of $\mc{A}_{\atoms}$ with advice $\beta$ (for a finite word $w$, $\delta_{\atoms}(w)$ is the state reached after reading $w \otimes \beta[:|w|]$). Let $\mc{L}_{\atoms}^n(q)$ be the partial language of $\mc{A}_{\atoms}$ in $q$ with advice $\beta[n:]$ (infinite words accepted starting in the state $q$ with advice $\beta[n:]$). We use similar notations $\delta_{\subseteq}$ and $\mc{L}_{\subseteq}^n$ for the automaton $\mc{A}_{\subseteq}$.

\begin{lemma}

\label{lem:comb}

Let $K_1:= 2 |\qi| +1$. For all $n \ge 0$ and every state $q \in \qs$, either $|\mc{L}_{\atoms}^{n+1}(q)| < K_1$ or $|\{v~|~|v| = n \text{ and } \delta_{\atoms}(v) = q\}| < K_1$.

\end{lemma}

\begin{proof} Let $B:= \{v~|~|v| = n \text{ and } \delta_{\atoms}(v) = q\}$.

Assume that you have $K_1$ distincts (finite) words $v_1 \dots v_{K_1}$ in $B$, and $K_1$ distincts (infinite) words $v'_1 \dots v'_{K_1}$ in $\mc{L}_{\atoms}^{n}(q)$. Define then $w_{i,j}:= v_i v'_j$ for $1 \le i,j \le K_1$. According to the definitions, we have $w_{i,j} \in \atoms$ (since $v_i$ leads to $q$ in $n$ steps and $v'_j$ is accepted starting in $q$ at time $n$).

Let $W:= \{w_{i,j}~|~1 \le i,j \le K_1\}$. Let $W' \subseteq \dom(\mc{S})$ be the set of elements encoded (as singletons) by the words of $W'$. Then $|W'| = (K_1)^2$. Let $P = \mc{P}(W') \subseteq \mc{P}^{f}(\dom(\mc{S}))$ and $C$ the set of (infinite) words encoding the elements of $P$ in the presentation of $P^f(\mc{S})$. Then $|C| = |P| = 2^{(K_1)^2}$.

For each $w \in C$ define:
\item
\begin{itemize}

\item $d_w: [1, K_1] \rightarrow \qi$ mapping $i$ to $\delta_{\subseteq}(v_i \otimes w[:n])$;

\item $f_w:  [1, K_1] \times \qi \rightarrow \{0,1\}$ with $f_w(j,q) = 1$ if and only if $v'_j \in \mc{L}^{n}_{\subseteq}(q)$.

\end{itemize}

Then remark that if $w_1, w_2 \in C$ and $d_{w_1} = d_{w_2}$ and $f_{w_1} = f_{w_2}$, then $w_1 = w_2$. Indeed, this means that exactly the same $v_{i,j}$ are in relation $w_1$ and $w_2$.

Now there are $|\qi|^{K_1} = 2^{\log_2(|\qi|) K_1} \le   2^{|\qi| K_1}$ different $d_w$ possibles, and $2^{|\qi| K_1}$ different $f_w$ possibles. This meaning that we can define at most $2^{2|\qi| K_1}$ elements in $C$, so we have $(K_1)^2 \le 2 |\qi| K_1 = (K_1 -1)K_1$ what brings a contradiction.

\end{proof}

Now we can forget the powerset and work only with the atoms and Lemma \ref{lem:comb}. The following result shows that the number of possible prefixes before the $\ind$ is bounded.

\begin{lemma}

\label{lem:ineq1}

Let $I_n := \{w[:n]~|~w \in \atoms \text{ and } \ind(w) = n\}$, then $\forall n\ge 0$, $|I_n| <K_{im}$.

\end{lemma}

\begin{proof} 

Assume that $|I_n| \ge K_{im}$, then remark that $K_{im} = |\qs| K_1$. Our assumption implies that there is a state $q \in \qs$ and $v_1 \dots v_{K_1} \in I_n$ distincts such that for each one $\delta_{\atoms}(v_i) = q$. This implies that $|\{v~|~|v| = n \text{ and } \delta_{\atoms}(v) = q\}| \ge K_1$, therefore by Lemma \ref{lem:comb} we have $|\mc{L}_{\atoms}^{n}(q)| < K_1$.

Note that there must be (at least) an infinite word $v_1 v' \in \atoms $ whose index is $n$ (by definition of $v_1$). In particular, $n$ is an important position and $\{w~|~v_1 w \in \atoms\} = \mc{L}_{\atoms}^{n}(q) $ has more than $K_{im} \ge K_1$ elements. Contradiction.

\end{proof}

Next, we show that the number of possible suffixes after the index is bounded.

\begin{lemma}

\label{lem:ineq2}

Let $J_n := \{w[n:]~|~w \in \atoms \text{ and } \ind(w) = n\}$, then $\forall n\ge 0$ we have $|J_n| \le K_{im} |\qs|  |\Sigma|$.

\end{lemma}

\begin{proof} Since $|J_n| \le |\Sigma| | B |$ where $B:= \{w[n+1:]~|~w \in \atoms \text{ and } \ind(w) = n\}$, we only need to bound the size of $B$ by $K_{im} |\qs|$.

For all $v' \in B$, there exist a state $q \in \qs$  and a word $v$ of size $n+1$ such that $\delta_{Atoms}(v) = q$, $v' \in \mc{L}_{\atoms}^{n+1}(q)$ and $\ind(vv') = n$ (by definition). In particular, this means that $n+1 \not \in \imp(vv')$. By definition of $\imp$, $|\{v''~|~vv'' \in \atoms'\}| \le K_{im}$. But $\mc{L}_{\atoms}^{n+1}(q) \subseteq \{v''~|~vv'' \in \atoms'\}$. Therefore $|\mc{L}_{\atoms}^{n+1}(q)| \le K_{im}$.

Since each $v' \in B$ comes from (at least) one such $\mc{L}_{\atoms}^{n+1}(q)$ (for a certain $q \in \qs$), we have shown $|B| \le K_{im} |\qs | $.

\end{proof}

Lemma \ref{lem:bound} follows directly from Lemmas \ref{lem:ineq1} et \ref{lem:ineq2}. Indeed, if $\ind(w) = n$, then $w[:n] \in I_n$ and $w[n:] \in J_n$, this meaning that $|Index^{-1}(n)| \le  K := K_{im} \times (K_{im} |Q_{\atoms}| |\Sigma|)$. 

\paragraph*{Back to homogeneous presentations}

\label{sec:conclusion}

We come back to the proof of Lemma \ref{lem:pro}. Recall that we want to show that $\mf{A}$ has an homogeneous $\reg^\infty[\beta]$-presentation. We have built a function $\ind: \atoms \rightarrow \mb{N}$ such that:

\begin{itemize}

\item $\{w \otimes 0^n \square^\omega~|~\ind(w) = n\} \in \wreg[\beta]$;

\item there is $K$ such that $|\ind^{-1}(n)| \le K$ for all $n$.

\end{itemize}

Let $1 \dots K$ be new letters, we can order the elements of each $\ind^{-1}(n)$ such that $ \{w \otimes k^n \square^\omega~|~w \text{ is the } k\text{-th word in } \ind^{-1}(n)\} \in \wreg[\beta]$. Indeed there is a regular well-ordering over finite words. This ordering defines an injective function $\code: \atoms \rightarrow \bigcup_{1 \le k \le K} k^*$ mapping $w$ to the unique $k^n$ such that $w \otimes k^n \square^\omega \in L$. Furthermore the language $ L:= \{k^n \square^\omega~|~\exists w \code(w) = k^n \} $ is in $ \wreg[\beta]$. The function $\eta : L \rightarrow A$  defined by $\eta(k^n \square^\omega) = \nu(\code^{-1}(k^n)) $ and is a bijection between $L$ and $A$. It is then easy to check that $\eta$ is the encoding function of an (injective) $\wreg[\beta]$-presentation of $\mf{A}$ over the language $L$. Since all the words of $L$ are of the form $k^n \square^\omega$, removing the padding symbols $\square$ from the encodings provides an homogeneous $\regi[\beta]$-presentation.

\section{Proof of Theorem \ref{theo:2WFT}}

\label{proof:theo:2WFT}

As evoked in the main body of this paper, we show that $\alpha \pc_{\MSOT} \beta$ implies $\alpha \pc_{\WFT} \beta$. The sketch of this proof is the following. We will first recall a result from \cite{alur2012}: if $\alpha \pc_{\MSOT} \beta$, then $\alpha$ can be computed from $\beta$ by a \emph{streaming string transducer}. We then transform such a transducer into a $\WFT$ when fixing its input. This result is not the same as Theorem \ref{theo:alur}, since we do not want to obtain $\omega$-lookaheads in the resulting transducer.

\subsection{Streaming string transducers}

Informally, a \emph{streaming string transducer} is a one-way transducer with a finite set $\mc{X}$ of registers that can store information and can be updated in a simple way. We call \emph{substitution} a mapping from $\mc{X}$ to $(\Gamma \uplus \mc{X})^*$. A substitution $\sigma$ is said to be \emph{copyless} if each variable appears at most once in the hole set $\{\sigma(x) ~|~x\in \mc{X}\}$. Let $\mc{S}_{\mc{X},\Gamma}$ the set of all copyless substitutions. Substitutions can be extended from $(\Gamma \uplus \mc{X})^*$ to $(\Gamma \uplus \mc{X})^*$ and thus composed.

\begin{example}

$x \mapsto x, y \mapsto x$ is not copyless, but $x \mapsto z, y \mapsto yx$ is so.

\end{example}

\begin{definition}[\cite{alur2012}] A (determinstic) \emph{streaming $\omega$-string transducer ($\SST$)} is an $8$-tuple $\mc{S} = ( Q, q_0, \Delta, \Gamma, \delta, \mc{X}, \lambda, F)$ where:

\begin{itemize}

\item $Q$ is the finite set of states with initial state $q_0$;

\item $\Delta$ (resp. $\Gamma$) is the input (resp. output) alphabet;

\item $\delta : Q \times \Delta \rightarrow Q$ is the (partial) transition function;

\item $\mc{X}$ is a finite set of registers (also called variables);

\item $\lambda : Q \times \Delta \rightarrow $ is the (partial) register update;

\item $F : \mc{P}(Q) \rightarrow \mc{X}^*$ is the (partial) output function such that for $P \in \dom(F)$ the string $F(P)$ is copyless of the form $x_1 \cdots x_n$ and for all $q,q' \in P$ with $q' = \delta(q,a)$ we have $\rho(q,a)(x_i) = x_i$ for all $i < n$ and $\lambda(q,a)(x_n) = x_n u$ for some $u \in \Gamma\cup \mc{X}^*$.

\end{itemize}

\end{definition}

If $P$ is the set of states that appear infinitely often in the run of  $\mc{S}$ on $\beta \in \Delta^\omega$, we want the value of $F(P) = x_1 \cdots x_n$ to be the (infinite) output string. This value has to change in a convergent way, therefore we force $\lambda(q,a)(x_i) = x_i$ and $\lambda(q,a)(x_n) = x_n u$ when the transition remains in $P$. This way, the register update only ``adds something in the end''.

More formally, a run $\rho$ of $\mc{S}$ on $\beta \in \Delta^\omega$ is a run of the underlying deterministic automaton $( Q, q_0, \Delta, \delta)$: for $k \ge 0$, $\rho(k)$ is the state reached after reading \mbox{$\alpha[:n]$}. We define inductively a sequence of \emph{ground} (from $\mc{X}$ to $\Gamma^*$) substitutions $(\sigma_k)_{k \ge 0}$: $\sigma_0(x) = \varepsilon$ for all $x$ and $\sigma_{k+1} = \sigma_k \circ \lambda(\rho(k),\alpha[k])$. If $P$ is the set of states appearing infinitely often in $\rho$ and $F(P) = x_1 \cdots x_n$, then necessarily $\lim_k \sigma_k(x_1 \cdots x_n)$ converges to a (finite or infinite) word. If finite, we make it infinite adding $\square^\omega$ in the end. It defines the output of $\mc{S}$ on $\beta$.

\begin{example} \label{ex:sst} Figure \ref{fig:2WFT} shows an $\SST$ that outputs the word $\widetilde{\alpha}$ on the input $\alpha \in (\Delta^*\#)^\omega$. Let $\mc{X} := \{x, \out\}$, and $F(\{q_0\}) = \out$; transitions are represented with the notation $[\textsf{input}~|~\textsf{substitution}]$ and $a$ stands for every symbol except $\#$. 

\begin{figure}[h!]
\begin{center}
      \begin{tikzpicture}[scale=0.6]
        \tikzstyle{etat} = [draw, fill=white];

        \node (dep) at (-2,0) {};
        \node[state] (l1) at (0,0) {$q_0$};
        \draw[->] (dep) edge[]  node [above] {} (l1);
        \draw[->] (l1) edge [loop above] node {$\Large \substack{a~|~x\mapsto ax,~\out \mapsto \out \\ \#~|~x\mapsto \varepsilon,~\out \mapsto \out x}$} (l1);
       \end{tikzpicture}

    \caption{\small \label{fig:2WFT} A simple $\SST$}
 \end{center}
\end{figure}

\end{example}

\begin{theorem}[\cite{alur2012}, Theorem 2]

\label{theo:sst}

$\SST$-definable functions between infinite strings are exactly $\MSOT$-definable functions.

\end{theorem}

We shall only need a weaker reformulation of this result, i.e. if $\alpha \pc_{\MSOT} \beta$ then $\alpha$ is computable from $\beta$ by an $\SST$. The model can even be simplified in that case.

\begin{definition} A \emph{simple $\SST$} is an $\SST$ with a distinguished register $\out \in \mc{X}$ such that  $\dom(F) = 2^Q$, and for all set of states $P$, $F(P) = \out$.

\end{definition}

\begin{example} The $\SST$ of Example \ref{ex:sst} is simple.

\end{example}

\begin{corollary}
\label{cor:ssst}
If $\alpha \pc_{\MSOT} \beta$ then $\alpha$ is the image of $\beta$ under some function realized by a simple $\SST$.
\end{corollary}

\begin{proof}

By applying Theorem \ref{theo:sst} (in its weaker reformulation), if $\alpha \pc_{\MSOT} \beta$ then $\alpha$ is the image of $\beta$ under some $\SST$ $\mc{S} := ( Q, q_0, \Delta, \Gamma, \delta, \mc{X}, \lambda, F)$. Let $\rho$ be the run on $\beta$, $(\sigma_k)_{k \ge 0}$ the associated sequence of ground substitutions, $P$ be the set of states that appear infinitely often in $\rho$. Necessarily $P \in \dom(F)$ thus $F(P) = x_1 \cdots x_n$. This sequence of registers was unpredictable when the input could change, but now it is fixed since $\beta$ is so. Let $K \ge 0$ such that for all $k \ge K$, $\rho(k) \in P$. For $x \in \mc{X}$ let $w_x := \sigma_K(x) \in \Gamma^*$. According to the remarks above, if $i < n$  then for all $k \ge K$,  $\sigma_k(x_i) = w_{x_i}$.

We define a simple $\SST$ $\mc{S}'$. The set of registers is $\mc{X} \uplus \{\out\}$ and the states are $\{0 \dots K-1\} \uplus P$. The graph of this transducer begins with a line of length $K$ on $\{0 \dots K-1\}$ with trivial updates of the registers, and in the last transition the value of $x \in \mc{X}$ is updated to $w_x$ and $\out$ is updated to $w_{x_1} \cdots w_{x_n}$. This transition leads to the state $\rho(K) \in P$ and then $\mc{S}'$ moves in $P$ like $\mc{S}$ does. The updates of $\out$ are defined like the updates of $x_n$. We define the output function of $\mc{S}'$ so that it fulfills the requirements of a simple $\SST$. It is easy to check that $\alpha$ is the image of $\beta$ under the function realized by $\mc{S}'$. Note that this transducer has no reason to preserve the output of $\mc{S}$ on other words. 

\end{proof}

\begin{remark} It follows from the definitions that in a simple $\SST$, $\lambda(p,a)(\out)$ (if defined) must be of the form $\out w$ with $w \in (\Gamma \uplus \mc{X})^*$ and no other register can use the value of $\out$.
\end{remark}

\subsection{Transforming simple $\SST$ into $\WFT$}

We first show how to transform a simple $\SST$ into a two-way transducer with a lookbehind feature. Informally, such a transducer has access to the state of a finite automaton reading a prefix of the input word.

\begin{definition}

A \emph{two-way transducer with lookbehind ($\WFT_b$)} is a two-way transducer $\mc{T} = (Q, q_0,  \Delta \uplus \{\vdash\}, \Gamma, \delta, \theta)$ together with a deterministic automaton $\mc{A} = (S, s_0, \Delta, \zeta)$ ($S$ set of states, $s_0 \in S$ initial state, $\zeta$ (partial) transition function), such that $\delta: Q \times (\Delta \uplus \{\vdash\}) \times S \rightarrow Q \times \{\triangleleft, \triangleright\}$ (partial function).

\end{definition}

When in state $q \in Q$ and position $n$ of the input $\beta \in \Delta^\omega$, the transition of the $\WFT_b$ is chosen as a function of $q$, $\beta[n]$ and $s:= \zeta(s_0,\beta[:n])$. Note that the access to $\zeta(s_0,\beta[:n])$ is purely an oracle and requires no effective run of $\mc{A}$. The definition of a run and the of output mechanism of a $\WFT_b$ is straightforward and similar to that of a $\WFT$.

\begin{lemma}

\label{lem:ssstwftb}

 If $\alpha$ is the image of $\beta$ under a simple $\SST$, then $\alpha$ is the image of $\beta$ under some $\WFT_b$.

\end{lemma}

\begin{proof}[Proof sketch]

Let $\mc{S}= ( S, s_0, \Delta, \Gamma, \zeta, \mc{X}, \lambda, \out)$ be a simple $\SST$ and $\beta \in \Delta^\omega$ such that $\alpha \in \Gamma^\omega$ is output by $\mc{S}$ on input $\beta$. We denote by $\rho$ the corresponding run and $(\sigma^k)_{k \ge 0}$ the sequence of ground substitutions. We are going to build a $\WFT_b$ $\mc{T}$  outputting $\alpha$ on $\beta$. Let $\mc{A} = (S, s_0, \Delta, \zeta)$ be the deterministic automaton of $\mc{S}$, this automaton will be the lookbehind of $\mc{T}$. In other words, we can assume that when in position $k+1$ of the input tape, $\mc{T}$ chooses its transition depending on $(\vdash\beta)[k+1] = \beta[k]$ and $\rho(k) = \zeta(s_0, \beta[:k])$. Hence the transition can also depend on $\lambda(\rho(k),\beta[k])$. 

In order to keep the proof readable, we shall provide a pseudocode describing the behavior of $\mc{T}$. The main issue is that since $\mc{T}$ has no registers, it cannot store unbounded information before outputting, as $\mc{S}$ used to do. Hence, we shall use the two-way moves of  $\mc{T}$ to compute ``recursively'' the value of the registers at each step, and output immediately what was added to $\out$. But the recursion procedure has to be quite subtle, since we cannot store a stack.

The pseudocode is given in Algorithm \ref{algo:algosst}, we now justify its correction, i.e. that it computes $\alpha$ on input $\beta$. Recall that the input tape contains in fact $\vdash \beta$, hence the first ``move right'' sends the reading head on $\beta[0]$.

 \newcommand{\proc}{\text{\texttt{process}}}
 \newcommand{\pos}{\text{\texttt{pos}}}
 \newcommand{\regs}{\text{\texttt{reg}}}

\begin{algorithm}
\SetKw{KwVar}{Variables.}
\SetKw{KwMem}{Finite-memory machine.}
\SetKwProg{Fn}{Function}{}{}

 \KwVar{\emph{All variables are global. $\pos \in \mb{N}$ denotes the current position on the input tape, implicitly updated at each move; $\regs \in \mc{X}$ is the register we are currently working on; $\proc \in (\mc{X} \cup \Gamma)^*$ is what remains to be output for \regs; $x \in \mc{X} \cup \Gamma $ will be used temporarily.}}
 
 \Fn{Next}{
  
			\eIf{$\proc \neq \varepsilon$}{
   
				$x \leftarrow \proc[0]$;
   
				$\proc \leftarrow \proc[1:]$;
				
				\eIf{$x \in \mc{X}$}{
   
					move left;
					
					$\regs \leftarrow x$;
					
					\eIf{$(\vdash\beta)[\pos] = \vdash$}{
					
						$\proc \leftarrow \varepsilon$;
					
					}{
					
						$\proc \leftarrow \lambda(\rho(\pos-1), (\vdash\beta)[\pos])(\regs)$;
					
					}

				}{
   					output $x$;
				}
   
			}{
			
				\eIf{$\regs = \out$}{
   
					break the inner while of the main program;
   
				}{
					move right;
					
					find the unique $x \in \mc{X}, w_1 \in (\Gamma \cup \mc{X})^*, w_2 \in (\Gamma \cup \mc{X})^*$ such that $\lambda(\rho(\pos-1), (\vdash\beta)[\pos])(x) = w_1 \regs w_2$;
					
   					$\regs \leftarrow x$;
					
					$\proc \leftarrow w_2$;
				
				}

				}
			}
					
	\While{true}{
	
	move right;
	
	$\regs \leftarrow \out$;
	
	$\proc \leftarrow \lambda(\rho(\pos-1), (\vdash\beta)[\pos])(\out)$;
	
	\While{true}{
	
	Next();
	}
	}
	
 \caption{\label{algo:algosst} the transducer $\mc{T}$}
\end{algorithm}

\begin{description}

\item[Algorithm \ref{algo:algosst} really describes a $\WFT_b$.] Variables only contain a bounded information, except \texttt{pos} but it is not directly used in the computations. Hence we describe a finite-memory machine. Conditions of the ``if'' depend on information which is available by a $\WFT_b$.

\item[Instruction ``find the unique''.] It is not clear that such $x, w_1, w_2$ exist and are unique. Their existence will be ensured at runtime. Uniqueness follows directly from the fact that the $\lambda$ substitutions are copyless. Note that this instruction is the key argument to make the procedure work without a stack of unbounded size.

\item[Specifications of Next.] The key invariant is the following. Let $k+1$, $x$ and $w$ be the values of $\pos$, $\regs$ and $\proc$ at a certain instant, such that $k \ge 0$ and $w$ is a suffix of $\lambda(\rho(k),\beta[k])(x)$. We claim that after a certain number of calls to Next(), we have $\pos = k+1$, $\regs = x$, $\proc = \varepsilon$, and during this time $\mc{T}$ has output $\sigma_k(w) \in \Gamma^*$. This result can by shown by induction on $(k, |w|)$ with lexicographical ordering.

\item[Main invariant.] Let $w_k \in (\mc{X} \cup \Gamma)^*$ such that $\lambda(\rho(k),\beta[k])(\out) = \out w_k$. Using what was done for Next we get that after the $(k+1)$-th main ``while'', $\mc{T}$ has output the string $\sigma_0(w_0) \cdots \sigma_k(w_k) \in \Gamma^*$.

\item[Full correction.] 
Since $\alpha = \lim_k \sigma_k(\out)$ and $\sigma_{k}(\out) = \sigma_0(w_0) \cdots \sigma_{k-1}(w_k)$, the previous invariant shows that $\mc{T}$ outputs $\alpha$.

\end{description}
\end{proof}

A run of this algorithm is detailed in Example \ref{ex:recurs} below.

\begin{example}

\label{ex:recurs}

Let $\mc{X} = \{x,y,\out\}$ and $\Gamma = \{a,b\}$. The last substitutions applied are written under the positions $k-1$, $k$ and $k+1$. We  want to output the last value of $xa$ added to $\out$ in position $k+1$.

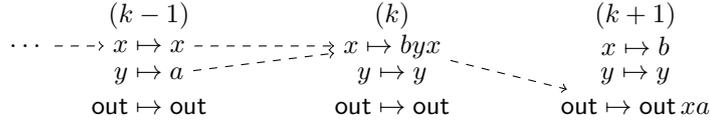
\vspace*{-0.2cm}
    \begin{center}
      \begin{tikzpicture}[scale=0.8]
        \tikzstyle{etat} = [draw, fill=white];
         \node (41) at (-5,1) {\dots};
        \node (30) at (-3,1.5) {$(k-1)$};
	\node (31) at (-3,1) {$x \mapsto x$};
	\node (32) at (-3,0.5) {$y \mapsto a$};
        \node (33) at (-3,0) {$\out \mapsto \out$};
        \node (20) at (1,1.5) {$(k)$};
	\node (21) at (1,1) {$x \mapsto b y x$};
	\node (22) at (1,0.5) {$y \mapsto y$};
        \node (23) at (1,0) {$\out \mapsto \out$}; 
        \node (10) at (5,1.5) {$(k+1)$};
	\node (11) at (5,1) {$x \mapsto b$};
	\node (12) at (5,0.5) {$y \mapsto y$};
        \node (13) at (5,0) {$\out \mapsto \out x a$}; 
        \draw[->,dashed] (41) edge[]  (31);
        \draw[->,dashed] (32) edge[]  (21);
        \draw[->,dashed] (31) edge[]  (21);
        \draw[->,dashed] (21) edge[]  (13);
       \end{tikzpicture}
    \end{center}
    
    \vspace*{-0.2cm}
    
    The $\WFT_b$ moves left from $(k+1)$ to $k$ to find $\sigma_{k+1}(x)$ (since $x$ is the first variable appearing in $x a$). It can already output $b$, then goes left to look for $\sigma_{k}(y)$. It outputs $a$ and reaches the end of a branch, so it moves right to $(k)$, keeping in memory that the last register was $y$, which only appears as a right member in $x \mapsto byx$. So the next value to output is $\sigma_{k}(x)$, we do it with the same procedure. When it has finally output the full value of $\sigma_{k+1}(x)$, it moves right and notice that the whole recursion process ends after outputting $a$.

\end{example}

\begin{remark} We could in fact show that if a function is realizable by a simple $\SST$, it can be realized by a $\WFT_b$ as well, but it is not useful in our context. This refined statement does not hold for $\SST$ in general.

\end{remark}

The next lemma shows that adding lookbehinds does not increase the expressiveness.

\begin{lemma}

\label{lem:wftbwft}

If $\alpha$ is the image of $\beta$ under a $\WFT_b$, then $\alpha$ is the image of $\beta$ under some $\WFT$ (without lookbehind).

\end{lemma}

\begin{proof}[Proof idea.] \cite{engelfriet2001} shows how to remove lookbehinds in the case of $\WFT$ over finite words. Since the ``behind'' only concerns a finite part of our infinite string, the adaptation is straightforward.
\end{proof}

We can now complete the proof of Theorem \ref{theo:2WFT}. If $\alpha \pc_{\MSOT} \beta$, then by Corollary \ref{cor:ssst} $\alpha$ is the image of $\beta$ under a simple $\SST$. By Lemma \ref{lem:ssstwftb} $\alpha$ is the image of $\beta$ under a $\WFT_b$. Finally Lemma \ref{lem:wftbwft} concludes that $\alpha \pc_{\WFT} \beta$.

\section{Proof of Theorem \ref{theo:prime}}

\label{proof:theo:prime}

We proceed in several steps to get that $\alpha \pc_{\WFT} \pi$ implies $\alpha \pc_{\DFT} \pi$. First (Lemma \ref{lem:10}), we show that a two-way computation on $\pi$ can be performed by a transducer that only changes its reading direction when seeing the letter $1$. We further prove (Lemma \ref{lem:big}) that it can be simulated by a one-way transduction in a ``bigger'' word, which is in the same $\DFT$ degree as $\pi$ (Lemma \ref{lem:pik}).

We assume there that $\alpha$ is not ultimately periodic. According to Lemma \ref{lem:wftper} $n \ge 0$, there is a moment when the transducer no longer goes before position $n$ on its input tape. Our constructions will refer implicitly to what happens ``far enough'' in the word.

\begin{lemma}

\label{lem:10}

If $\alpha \pc_{\WFT} \pi$, the transformation can be performed by a transducer whose states $Q$ are partitioned in two sets $Q^{\triangleright}$ and $Q^{\triangleleft}$ such that the following holds for the transition function $\delta$. For all $q \in Q^{\triangleright}$ (resp. $Q^{\triangleleft}$) $\delta(q,0) = (q',\triangleright)$ (resp. $ (q',\triangleleft)$); for all $q \in Q$ and $a \in \{0,1\}$, $\delta(q,a) = (q',\triangleright)$ (resp. $ (q',\triangleleft)$) implies $q' \in Q^{\triangleright}$ (resp. $q' \in Q^{\triangleleft}$).

\end{lemma}

\begin{proof}[Proof sktech.]

Let $\mc{T}$ be a $N$-states transducer performing the transformation. We study how $\mc{T}$ copes with the $10^n1$ blocks of $\pi$. Assume $\mc{T}$ enters the block from the left side and goes right after reading the $1$. We consider the (two-way) run staying in $0^n$, before the next visit of a letter $1$. Since $\mc{T}$ can only see $0$, it is finally (after a most $N$ steps) caught in a loop (of size at most $N$). By loop we mean a two-way loop, between by two configurations sharing the same state. Two cases may occur.
\item
\begin{itemize}
\item The next $1$ visited is the left one: $\mc{T}$ comes back to its previous position. In that case, the run in $0^n$ cannot be longer than $N + N^2$ (due to the loop), else it should go ``right''.

\item The next $1$ visited is the right one: $\mc{T}$ went through to block $0^n$.
\end{itemize}

For $n$ large enough, the occurring case does not depends on $n$, but only on the state when entering the block. It is not hard to derive a formal construction from the previous remark: the first case can be hardcoded without moving (since the run is bounded) , and the second one can be simulated in a one-way manner (simulate a bounded loop). Adapt consequently the output function. We get a new transducer that does not change its reading direction in blocks of $0$.

\end{proof}

\begin{remark} This lemma remains valid when replacing $\pi$ by any binary sequence where the gap between two consecutive $1$ goes to the infinity.

\end{remark}

 Let $\pi^k = \prod_{n=0}^{\infty} (0^n1)^k$ for $k \ge 1$.

\begin{lemma}
\label{lem:big}

$\alpha \pc_{\WFT} \pi$, there is $k \ge 1$ such that $\alpha \pc_{\DFT} \pi^k$.

\end{lemma}

\begin{proof}[Proof sketch.]
Let $\mc{T}$ be the $N$-states transducer obtained by Lemma \ref{lem:10}. Let $p_n$ for $n \ge 0$ be the position of the $(n+1)$-th letter $1$ in $\pi$, namely $p_n = \frac{n(n+1)}{2}$. A non-trivial run on $\pi$ can be decomposed in finite runs $R_n$ between the last visit in $p_n$ and the last visit in $p_{n+1}$.

We claim that there is $c \ge 1$ such that for all $n$ large enough, $R_n$ never visits position $p_{n+c}$. Thus $R_n$ is contained between $p_n$ and $p_{n+c}$, so in $10^n 1 \cdots 0^{n+c} 1$. It follows that $R_n$ can be simulated by a one-way run on $(0^n 1)^{cN}$, by unfolding of the two-way moves (we use here the previous lemma). Putting everything together, we get that $\alpha \pc_{\DFT} \pi^{cN}$.

Let us now prove what we claimed. The idea is that if $\mc{T}$ goes arbitrarily far in $R_n$, it must be caught in a (two-way) loop and thus cannot come back to $p_{n+1}$. This is not totally trivial, because the $10^i1$ blocks are of increasing size in the word, what might ``break'' the loop. Let $\mc{A}$ be the underlying (deterministic) automaton of $\mc{T}$, $\delta$ its transition function and $m$ the least common multiple of the cycles labelled with $0$ in the graph of $\mc{A}$. For all $i$ large enough and congruent modulo $m$, the state $\delta(q,10^i1)$ is independent from $i$, since we must finish every cycle. Thus if $\mc{A}$ goes ``too far'', there exists in $R_n$ two configurations, in the same state and in positions $p_l$ and $p_{l'}$ with $l<l'$, before equally congruent blocks. This initiates a two-way loop that makes $\mc{A}$ unable to go to $p_{n+1}$.
\end{proof}

\begin{lemma}
\label{lem:pik}

For all $k \ge 1$, $ \pi$, $\pi^k \pc_{\DFT} \pi$.
\end{lemma}

\begin{proof}

We rewrite $\pi$ as $\prod_{n=0}^{\infty} \prod_{j=0}^{k-1}0^{kn + j}1$. A one-way transducer outputting $(0^n 1)^k$ when reading $\prod_{j=0}^{k-1}0^{kn + j}1$ can be built.

\end{proof}

Lemmas \ref{lem:big} and \ref{lem:pik} give $\alpha \pc_{\DFT} \pi$ and hence conclude the proof.

\end{document}